\documentclass{article}%
\usepackage{amsmath}
\usepackage{amsfonts}
\usepackage{amssymb}
\usepackage{graphicx}%
\providecommand{\U}[1]{\protect\rule{.1in}{.1in}}
\newtheorem{theorem}{Theorem}
\newtheorem{acknowledgement}[theorem]{Acknowledgement}

\newtheorem{definition}[theorem]{Definition}

\newtheorem{proposition}[theorem]{Proposition}
\newtheorem{remark}[theorem]{Remark}

\newenvironment{proof}[1][Proof]{\noindent\textbf{#1.} }{\ \rule{0.5em}{0.5em}}
\begin{document}

\title{A Clifford Bundle Approach to the Wave Equation of a Spin $1/2$ Fermion in the
de Sitter Manifold.}
\author{{\footnotesize W. A. Rodrigues Jr.}$^{(1)}${\footnotesize , S. A.
Wainer}$^{{\footnotesize (1)}}${\footnotesize , M. Rivera-Tapia}$^{(2)}%
${\footnotesize , E. A. Notte-Cuello}$^{(3)}${\footnotesize and I.
Kondrashuk}$^{(4)}$\\$^{(1)}${\footnotesize Institute of Mathematics, Statistics and Scientific
Computation}\\{\footnotesize IMECC-UNICAMP}\\{\footnotesize walrod@ime.unicamp.br~~~~samuelwainer@ime.unicamp.br}\\$^{(2)}${\footnotesize Departamento de F\'{\i}sica, Universidad de La Serena,
La Serena-Chile.}\\{\footnotesize marivera@userena.cl}\\$^{(3)}${\footnotesize Departamento de Matematicas, Universidad de La Serena,
La Serena-Chile.}\\{\footnotesize enotte@userena.cl}\\$^{(4)}${\footnotesize Grupo de Matem\'{a}tica Aplicada, Departamento de
Ciencias B\'{a}sicas,}\\{\footnotesize Universidad del B\'{\i}o-B\'{\i}o, Campus Fernando May, Casilla
447, Chill\'{a}n, Chile\ }\\{\footnotesize igor.kondrashuk@ubiobio.cl}}
\date{July 19 2015}
\maketitle

\begin{abstract}
In this paper we give a Clifford bundle motivated approach to the wave
equation of a free spin $1/2$ fermion in the de Sitter manifold, a brane with
topology $M=\mathrm{S0}(4,1)/\mathrm{S0}(3,1)$ living in the bulk spacetime
$\mathbb{R}^{4,1}=(\mathring{M}=\mathbb{R}^{5},\boldsymbol{\mathring{g}})$ and
equipped with a metric field $\boldsymbol{g}:\boldsymbol{=}-\boldsymbol{i}%
^{\ast}\boldsymbol{\mathring{g}}$ with $\boldsymbol{i}:M\rightarrow
\mathring{M}$ being the inclusion map. To obtain the analog of Dirac equation
in Minkowski spacetime in the structure $\mathring{M}$ we appropriately
factorize the two Casimir invariants $C_{1}$ and $C_{2}$ of the Lie algebra of
the de Sitter group using the constraint given in the linearization of $C_{2}$
as input to linearize $C_{1}$. In this way we obtain an equation that we
called \textbf{DHESS1, }which in previous studies by other authors was simply
postulated.$.$Next we derive a wave equation (called \textbf{DHESS2}) for a
free spin $1/2$ fermion in the de Sitter manifold using a heuristic argument
which is an obvious generalization of a heuristic argument (described in
detail in Appendix D) permitting a derivation of the Dirac equation in
Minkowski spacetime and which shows that such famous equation express nothing
more than the fact that the momentum of a free particle is a constant vector
field over timelike \ integral curves of a given velocity field. It is a
remarkable fact that \textbf{DHESS1 }and \textbf{DHESS2}\ coincide. One of the
main ingredients in our paper is the use of the concept of Dirac-Hestenes
spinor fields. Appendices B and C recall this concept and its relation with
covariant Dirac spinor fields usually used by physicists.

\end{abstract}


\section{Introduction}

The Dirac equation (\textbf{DE}) in a Minkowski spacetime can be obtained
using Dirac's original procedure through a linearization of $C_{1}-m^{2}=0$
(where~$C_{1}$ is the first Casimir invariant of the enveloping algebra of the
Poincar\'{e} group) and its application to covariant spinor fields (sections
of $P_{\mathrm{Spin}_{1,3}^{e}}\times_{\mu}\mathbb{C}^{4}$, see Appendices B
and C).\ In the Appendix D using the Clifford and spin-Clifford bundles
formalism\footnote{This means the Clifford and spin-Clifford bundles formalism
as developed in \cite{rc2007}. We use the notations of that book and the
reader is invited to consult the book if he needs to improve his knowledge in
order to be able to follow all calculations of the present article.} and an
almost trivial heuristic argument we present a derivation of an equivalent
equation\ to \textbf{DE }which is called the Dirac-Hestenes equation
(\textbf{DHE}). Our derivation makes clear the fact that the \textbf{DE} (or
the equivalent \textbf{DHE}) express nothing more than the fact that a free
spin $1/2$ particle moves with a constant velocity in Minkowski spacetime
following an integral line of a well defined velocity field\footnote{The way
in which the intrinsic spin of the particle is treated in this formalism has
ben carefully discussed in \cite{rvp1996}.}. This observation is a crucial one
for the main objective of this paper, the one of writing wave equations for a
free spin $1/2$ moving in a de Sitter manifold equipped with a metric field
inherited from a bulk spacetime $\mathbb{R}^{4,1}$ (see Section 2). It is
intuitive (given the topology of the de Sitter manifold) that such a motion
must happen with a constant angular momentum\footnote{On this respect see also
section XIV.3 of \cite{arc}.} and as we will see a heuristic deduction of a
Dirac-Hestenes like equation in this case results identical from the one which
we get if we linearize $C_{1}-m^{2}=0$ (where $C_{1}$ is the first Casimir
invariant of the enveloping algebra of the Lie algebra of the de Sitter group)
taking into account a a constraint coming from the linearization of \ $C_{2}$,
the second Casimir invariant of the enveloping algebra of the Lie algebra of
the de Sitter group

To be more precise, in Sections 3 \ and 4 we will present two
Dirac-Hestenes\ like equations for a spin $1/2$ fermion field\footnote{Called
in what follows a Dirac-Hestenes spinor field and denoted \textbf{DHSF}$.$}
$\phi$ living in de Sitter manifold equipped with a metric field
$\boldsymbol{g}$ (see Section 2), which will be abbreviate as \textbf{DHESS1}
and \textbf{DHESS2}. The \textbf{DHESS1} will be obtained by linearizing the
first Casimir operator $C_{1}$ using a constraint imposed on the \textbf{DHSF}
arising from the linearization of $C_{2}$. On the other hand \textbf{DHESS2}
will be obtained by a physically and heuristically derivation resulting by
simply imposing that the motion of a free particle in the de Sitter manifold
is described by a constant angular momentum $2$-form as seem by an
hypothetical observer living in the bulk spacetime $\mathbb{R}^{4,1}$. Of
course, as we are going to see the heuristic derivation is only possible using
the Clifford bundle formalism. It is a remarkable result that \textbf{DHESS1
}and \textbf{DHESS2} coincide and moreover translation of those equations in
the covariant spinor field formalism gives a first order partial differential
equation (which is equivalent to the one first postulated by Dirac
\cite{dirac}). It will be shown that \textbf{DHESS1 (}and
thus\ \textbf{DHESS2}) reduces to the Dirac-Hestenes equation (\textbf{DHE)}
in Minkowski spacetime when $\ell\rightarrow\infty$, where $\ell$ is the
radius of the de Sitter manifold.

In Section 5 \ we study the limit of \textbf{DHESS1} and \textbf{DHESS2 }when
$\ell\rightarrow\infty$ ($\ell$ being the radius of the de Sitter
manifold)showing that it gives the Dirac-Hestenes equation in Minkowski spacetime.

In Section 6 we present our conclusions, comparing our results with others
already appearing in the literature .We claim that our approach reveals
details of subject that are completely hidden in the usual matrix approach to
the subject\footnote{See, \cite{dirac,gursey}.} where Dirac fields are seen as
mappings $\Phi:M\rightarrow\mathbb{C}^{4}$, in particular the nature of the
object $\boldsymbol{\lambda}$ (Eq.(\ref{c19aa}) appearing in the linearization
of $C_{1}$ and related (but not equal) the mass of the particle.

The paper has four Appendices. Appendix A recalls the Lie algebra and Casimir
invariants of the de Sitter group. Appendices B and C\ have been written for
the reader's convenience since the subject is not well known to physicists.
Those appendices recall the main definitions and properties of the theory of
\textbf{DHSF} necessary for a complete intelligibility of the paper. Finally,
Appendix D presents a heuristic derivation of Dirac equation in Minkowski
spacetime that served as inspiration for the theory presented in the main text.

\section{The Lorentzian de Sitter $M^{dSL}$ Structure and its (Projective)
Conformal Representation}

Let $SO(4,1)$ and $SO(3,1)$ be respectively the special pseudo-orthogonal
groups in the structures $\mathbb{R}^{4,1}=\{\mathring{M}=\mathbb{R}%
^{5},\boldsymbol{\mathring{g}}\}$ and $\mathbb{R}^{3,1}=\{\mathbb{R}%
^{4},-\boldsymbol{\eta}\}$ where $\boldsymbol{\mathring{g}}$ is a metric of
signature $(4,1)$ and $\boldsymbol{\eta}$ a metric of signature $(1,3)$.\ The
manifold $M=SO(4,1)/SO(3,1)$ will be called the \emph{de Sitter manifold}.
Since
\begin{equation}
M=SO(4,1)/SO(3,1)\approx SO(1,4)/SO(1,3)\approx\mathbb{R\times}S^{3} \label{0}%
\end{equation}
this manifold can be viewed as a brane \cite{rw2014} (a submanifold) in the
structure $\mathbb{R}^{4,1}$. In General Relativity studies it is introduced a
Lorentzian spacetime, i.e., the structure $M^{dSL}=(M=\mathbb{R\times}%
S^{3},\boldsymbol{g},\boldsymbol{D},\tau_{\boldsymbol{g}},\uparrow)$
called\emph{ Lorentzian de Sitter spacetime structure\footnote{It is a vacuum
solution of Einstein equation with a cosmological constant term. We are not
going to use this structure in this paper.}} where if $\boldsymbol{\iota
}:\mathbb{R\times}S^{3}\rightarrow\mathbb{R}^{5}$ is the inclusion mapping,
$\boldsymbol{g}:=-\iota^{\ast}\boldsymbol{\mathring{g}}$ and $\boldsymbol{D}$
is the parallel projection on $M$ of the pseudo Euclidian metric compatible
connection in $\mathbb{R}^{4,1}$ (details in \cite{rw2015}). As well known,
$M^{dSL}$ is a spacetime of constant Riemannian curvature. It has ten Killing
vector fields. The Killing vector fields are the generators of infinitesimal
actions of the group $SO(4,1)$ (called the de Sitter group) in
$M=\mathbb{R\times}S^{3}\approx SO(4,1)/SO(3,1)$. The group $SO(4,1)$ acts
transitively\footnote{A group $G$ of transformations in a manifold $M$
($\sigma:G\times M\rightarrow M$ by $(g,x)\mapsto\sigma(g,x)$) is said to act
transitively on $M$ if for arbitraries $x,y\in M$ there exists $g\in G$ such
that $\sigma(g,x)=y$.} in $SO(4,1)/SO(3,1)$, which is thus a homogeneous space
(for $SO(4,1)$).

We now give a description of the manifold $\mathbb{R\times}S^{3}$ as a
pseudo-sphere (a submanifold) of radius $\ell$ of the pseudo Euclidean space
$\mathbb{R}^{4,1}=\{\mathbb{R}^{5},\boldsymbol{\mathring{g}}\}$. If
$(X^{1},X^{2},X^{3},X^{4},X^{0})$ are the global orthogonal coordinates of
$\mathbb{R}^{4,1}$, then the equation representing the pseudo sphere is%
\begin{equation}
(X^{1})^{2}+(X^{2})^{2}+(X^{3})^{2}+(X^{4})^{2}-(X^{0})^{2}=\ell^{2}.
\label{ds4}%
\end{equation}

Introducing projective \emph{conformal} coordinates $\{x^{\mu}\}$ by
projecting the points of $\mathbb{R\times}S^{3}$ from the \textquotedblleft
north-pole\textquotedblright\ to a plane tangent to the \textquotedblleft
south pole\textquotedblright\ we see immediately that $\{x^{\mu}\}$
covers\ all $\mathbb{R\times}S^{3}$ except the \textquotedblleft
north-pole\textquotedblright. We have
\begin{equation}
X^{\mu}=\Omega x^{\mu},~~~X^{4}=-\ell\Omega\left(  1+\frac{\sigma^{2}}%
{4\ell^{2}}\right)  \label{ds4a}%
\end{equation}
with
\begin{equation}
\Omega=\left(  1-\frac{\sigma^{2}}{4\ell^{2}}\right)  ^{-1},~~~\sigma^{2}%
=\eta_{\mu\nu}x^{\mu}x^{\nu} \label{ds2}%
\end{equation}
and we immediately find that%

\begin{equation}
\boldsymbol{g}:=-\boldsymbol{\iota}^{\ast}\boldsymbol{\mathring{g}}=\Omega
^{2}\eta_{\mu\nu}dx^{\mu}\otimes dx^{\nu}, \label{ds1}%
\end{equation}
and the matrix with entries $\eta_{\mu\nu}$ is the diagonal matrix
$\mathrm{diag}(1,-1,-1,-1)$.

Since the north pole of the pseudo sphere is not covered by the coordinate
functions we see that (omitting two dimensions) the region of the spacetime as
seem by an observer living the south pole is the region inside the so called
absolute of \emph{Cayley-Klein} of equation
\begin{equation}
t^{2}-\mathbf{x}^{2}=4\ell^{2}. \label{ds5}%
\end{equation}
%

\begin{figure}[ptb]%
\centering
\includegraphics[
height=4.1009in,
width=4.4538in
]%
{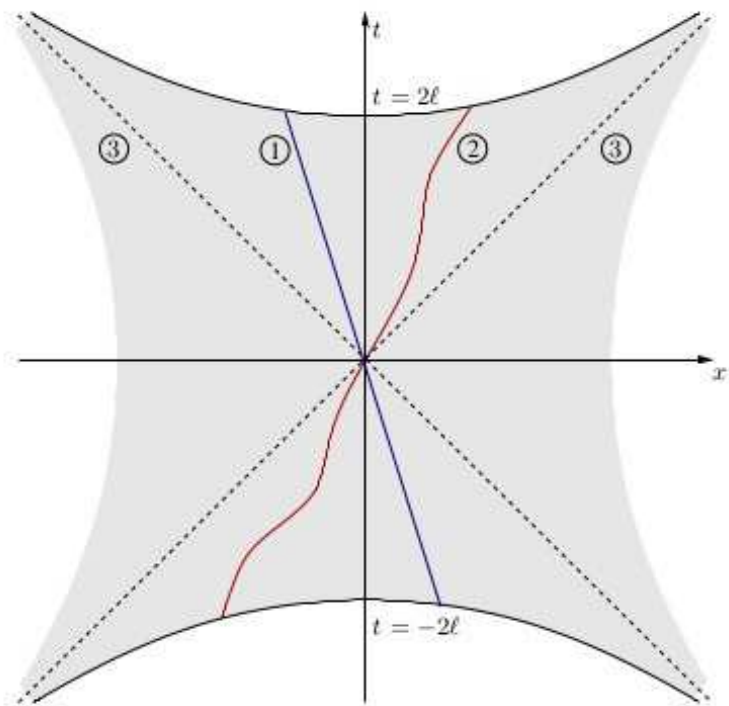}%
\caption{Projective conformal representation of de Sitter spacetime. Note that
the \textquotedblleft observer\textquotedblright\ spacetime is the interior of
the Cayley-Klein absolute $t^{2}-\mathbf{x}^{2}=4\ell^{2}$. }%
\end{figure}

In Figure 1 we can see that all timelike curves (2) and lightlike (1) starts
in the \textquotedblleft past horizon\textquotedblright\ and end on the
\textquotedblleft future horizon\textquotedblright. More details in
\cite{rw2015}.

\section{Linearization of the Casimir Invariants of the spin$_{4,1}^{e}$ Lie
algebra}

The classical angular momentum biform of a free particle following a
\textquotedblleft timelike\textquotedblright\ curve $\sigma$ with momentum
$1$-form $\boldsymbol{p}$ in the structure $\mathbb{R}^{4,1}$ is%
\begin{equation}
\boldsymbol{l=x}\wedge\boldsymbol{p}, \label{c0}%
\end{equation}
where
\begin{equation}
\boldsymbol{x=X}^{A}\boldsymbol{\mathring{E}}_{A},~~~~~\boldsymbol{p}%
=P_{A}\boldsymbol{\mathring{E}}^{A} \label{c3}%
\end{equation}
are respectively the position 1-form and the momentum of the free particle.
Moreover, $\{\boldsymbol{\mathring{E}}^{A}=dX^{A}\}$ is an orthonormal cobasis
of $T^{\ast}\mathring{M}$ dual to the orthonormal basis
$\{\boldsymbol{\mathring{e}}_{A}=\frac{\partial}{\partial X^{A}}\}$ of
$T\mathring{M}$. and $\{\boldsymbol{\mathring{E}}_{A}\}$ is an orthonormal
cobasis of $T^{\ast}\mathring{M},$ called the reciprocal basis of
$\{\boldsymbol{\mathring{E}}^{A}\}$ and it is $\mathtt{\mathring{g}%
}(\boldsymbol{\mathring{E}}^{A},\boldsymbol{\mathring{E}}_{B})=\delta_{B}^{A}$
where\footnote{The matrix with entries $\eta_{AB}$ is the diagonal matrix
$\mathrm{diag}(1,1,1,1,-1)$}%
\begin{equation}
\mathtt{\mathring{g}}=\eta^{AB}\boldsymbol{\mathring{e}}_{A}\otimes
\boldsymbol{\mathring{e}}_{B} \label{c000}%
\end{equation}
is the metric for $T^{\ast}\mathring{M}$. If
\begin{equation}
\boldsymbol{\mathring{g}}=\eta_{AB}\boldsymbol{\mathring{E}}^{A}%
\otimes\boldsymbol{\mathring{E}}^{B} \label{c00}%
\end{equation}
is the metric of $T\mathring{M}$, it is $\eta^{AC}\eta_{CB}=\delta_{B}^{A}$.
We have
\begin{equation}
\boldsymbol{l}=\frac{1}{2}L_{AB}\boldsymbol{\mathring{E}}^{A}\wedge
\boldsymbol{\mathring{E}}^{B}=\frac{1}{2}L_{AB}\boldsymbol{\mathring{E}}%
^{A}\boldsymbol{\mathring{E}}^{B} \label{c1}%
\end{equation}
with
\begin{equation}
L_{AB}=\eta_{AC}X^{C}P_{B}-\eta_{BC}X^{C}P_{A} \label{c2}%
\end{equation}

\begin{remark}
It is quite obvious that for a classical particle living in de Sitter
spacetime and following a timelike worldline $\sigma$ parametrized by proper
time $\tau$ if we write \emph{\cite{roldao1,roldao2}}%
\begin{equation}
\boldsymbol{x}=X^{A}(\tau)\boldsymbol{\mathring{E}}_{A}\text{,~~~~}%
\boldsymbol{p}=m\frac{dX^{A}(\tau)}{d\tau}\boldsymbol{\mathring{E}}_{A}
\label{C22}%
\end{equation}
it is \emph{(}since $\boldsymbol{x\cdot x=\mathring{g}}(\boldsymbol{x}%
,\boldsymbol{x})=\ell^{2}$\emph{)}%
\begin{equation}
\boldsymbol{x\cdot p}=0. \label{c222}%
\end{equation}
Thus,
\begin{equation}
\boldsymbol{xp=x}\wedge\boldsymbol{p} \label{c223}%
\end{equation}
and as a consequence
\begin{equation}
\boldsymbol{l}^{2}=\boldsymbol{xpxp}=-\boldsymbol{pxxp}=-\ell^{2}%
\boldsymbol{p}^{2} \label{c224}%
\end{equation}
which implies that
\begin{equation}
\boldsymbol{l\wedge l}=0 \label{c225}%
\end{equation}
and thus
\begin{equation}
\boldsymbol{l}^{2}=\boldsymbol{l\lrcorner l}\text{.} \label{c226}%
\end{equation}
As we re going to see the classical condition given by \emph{Eq.(\ref{c225})}
cannot be assumed in quantum theory where the classical angular momentum is
substituted by a quantum angular momentum operator.
\end{remark}

So, to continue we define $\mathcal{\mathring{H}}$, the Hilbert space of a one
quantum spin $1/2$ particle living $\mathring{M}$ as the set of all square
integrable mappings\footnote{By square integrable we mean that $\int\psi
\cdot\psi\tau_{\boldsymbol{\mathring{g}}}=1$.}%
\begin{equation}
\phi\in\sec\mathcal{C\ell}^{0}\mathcal{(}\mathring{M},\mathtt{\mathring{g}%
}\mathcal{)} \label{c4}%
\end{equation}
called representatives in $\mathcal{C\ell}^{0}\mathcal{(}\mathring
{M},\mathtt{\mathring{g}}\mathcal{)}$ relative to a spin coframe of
(\textbf{DHSF})\footnote{See details in Appendix B and C.} \cite{hs1984}.

The quantum angular momentum operator $\mathbf{L}\in\mathcal{L(\mathring{H})}$
is
\begin{equation}
\mathbf{L:}=\frac{1}{2}\boldsymbol{\mathring{E}}^{A}\boldsymbol{\mathring{E}%
}^{B}\mathbf{L}_{AB} \label{c6}%
\end{equation}
where
\begin{equation}
\mathbf{L}_{AB}=\eta_{AC}X^{C}\mathbf{P}_{B}-\eta_{BC}X^{C}\mathbf{P}_{A}
\label{c7}%
\end{equation}
with $\mathbf{P}_{A}\in\mathcal{L(\mathring{H})}$ defined\footnote{The
definition of the operators $\mathbf{P}_{A}$ acting on tangent spinor fields
to the de Sitter manifold is given below in Eq.(\ref{c36}).} by%
\begin{equation}
\mathbf{P}_{A}\phi:=\partial_{A}\phi\boldsymbol{\mathring{E}}^{2}%
\boldsymbol{\mathring{E}}^{1} \label{c7a}%
\end{equation}

Now\footnote{The action of $\mathbf{L}^{2}$ on $\phi$ \ (in analogy to the
square of the Dirac operator) is defined by
\par
$\mathbf{L}^{2}\phi:=\mathbf{L}(\mathbf{L}\phi)=(\mathbf{L}\lrcorner
\mathbf{L)}\phi+(\mathbf{L}\wedge\mathbf{L)}\phi.$},
\begin{equation}
\mathbf{L}^{2}=\mathbf{L\lrcorner L+L\wedge L.} \label{c9}%
\end{equation}

$\mathbf{L\lrcorner L}$ is clearly a scalar invariant under the action of
$\mathrm{Spin}_{4,1}^{e}$ group and it is:%
\begin{align}
\mathbf{L\lrcorner L}  &  =\frac{1}{4}L_{AB}L^{CD}(\boldsymbol{\mathring{E}%
}^{A}\wedge\boldsymbol{\mathring{E}}^{B})\lrcorner(\boldsymbol{\mathring{E}%
}_{C}\wedge\boldsymbol{\mathring{E}}_{D})\nonumber\\
&  =\frac{1}{4}L_{AB}L^{CD}\boldsymbol{\mathring{E}}^{A}\lrcorner
(\boldsymbol{\mathring{E}}^{B}\lrcorner(\boldsymbol{\mathring{E}}_{C}%
\wedge\boldsymbol{\mathring{E}}_{D}))\nonumber\\
&  =\frac{1}{4}L_{AB}L^{CD}\boldsymbol{\mathring{E}}^{A}\lrcorner(\delta
_{C}^{B}\boldsymbol{\mathring{E}}_{D}-\delta_{D}^{B}\boldsymbol{\mathring{E}%
}_{C})\nonumber\\
&  =\frac{1}{4}L_{AB}L^{CD}(\delta_{C}^{B}\delta_{D}^{A}-\delta_{D}^{B}%
\delta_{C}^{A})\nonumber\\
&  =-\frac{1}{2}L_{AB}L^{AB}. \label{c10}%
\end{align}
The first Casimir operator of the Lie algebra\ $\mathrm{spin}_{4,1}^{e}$\ is
defined by
\begin{equation}
C_{1}=\frac{1}{\ell^{2}}\mathbf{L\lrcorner L}=-\frac{1}{\ell^{2}%
}\mathbf{L\cdot L=}-\frac{1}{2\ell^{2}}L_{AB}L^{AB}=m^{2}, \label{c11}%
\end{equation}
with $m^{2}\in\mathbb{R}$.

We have the

\begin{proposition}
\label{casimir2}Call
\begin{equation}
\mathbf{W:=}\star\frac{1}{8\ell}(\mathbf{L\wedge L})=\frac{1}{8\ell
}(\mathbf{L\wedge L})\lrcorner\tau. \label{c11a}%
\end{equation}
Then,%

\begin{align}
\mathbf{W\cdot W}  &  =\mathbf{WW}=-\frac{1}{64\ell^{2}}(\mathbf{L\wedge
L})\lrcorner(\mathbf{L\wedge L})=-\frac{1}{64\ell^{2}}(\mathbf{L\wedge
L})\cdot(\mathbf{L\wedge L})\nonumber\\
&  =-\frac{1}{64\ell^{2}}(\mathbf{L\wedge L})(\mathbf{L\wedge L}). \label{c12}%
\end{align}

\end{proposition}

\begin{proof}
Recalling the identity \footnote{See \cite{rc2007}, page 33.}
\begin{gather}
A_{l}\lrcorner\star B_{s}=B_{s}\lrcorner\star A_{l}~~\text{ if }%
l+s=n=\dim\mathring{M},\nonumber\\
A_{l}\in\sec%
{\textstyle\bigwedge\nolimits^{l}}
T^{\ast}\mathring{M}\hookrightarrow\mathcal{C\ell(}\mathring{M}%
,\mathtt{\mathring{g}}\mathcal{)~~}\text{ }B_{s}\in\sec%
{\textstyle\bigwedge\nolimits^{s}}
T^{\ast}\mathring{M}\hookrightarrow\mathcal{C\ell(}\mathring{M}%
,\mathtt{\mathring{g}}\mathcal{)} \label{C13}%
\end{gather}
and taking also into account that
\begin{equation}
\star^{-1}A_{1}=-\star A_{1} \label{C14}%
\end{equation}
we can write
\begin{align}
(\mathbf{L\wedge L})\lrcorner(\mathbf{L\wedge L})  &  =(\mathbf{L\wedge
L})\lrcorner\star^{-1}\star(\mathbf{L\wedge L})\nonumber\\
-(\mathbf{L\wedge L})\lrcorner\star\star(\mathbf{L\wedge L})  &
=-\star(\mathbf{L\wedge L})\lrcorner\star(\mathbf{L\wedge L})\nonumber\\
&  =-((\mathbf{L\wedge L})\lrcorner\tau)\lrcorner((\mathbf{L\wedge
L})\lrcorner\tau)\nonumber\\
&  =-64\ell^{2}\mathbf{W\cdot W.} \label{c15}%
\end{align}
On the other hand we have that
\begin{equation}
(\star^{-1}\mathbf{W)(\star^{-1}\mathbf{W})}=(\star\mathbf{W)(\star
\mathbf{W})=W}\tau\tau\mathbf{W=-WW=-W\cdot W} \label{c15a}%
\end{equation}
and since $\star^{-1}\mathbf{W=L\wedge L}$ we have from Eqs. (\ref{c15}) and
(\ref{c15a}) that
\begin{equation}
(\mathbf{L\wedge L})(\mathbf{L\wedge L})=-64\ell^{2}\mathbf{W\cdot
W}=(\mathbf{L\wedge L})\lrcorner(\mathbf{L\wedge L}) \label{c15b}%
\end{equation}
which completes the proof.
\end{proof}

\begin{remark}
The second Casimir invariant of \textrm{spin}$_{4,1}^{e}$ is defined by
\begin{equation}
C_{2}=\mathbf{W\cdot W}=-\frac{1}{64\ell^{2}}(\mathbf{L\wedge L}%
)(\mathbf{L\wedge L})=-m^{2}s(s+1), \label{c16}%
\end{equation}
where $s=0,1/2,1,3/2,...$ It is thus quite obvious that contrary to the
classical case the operator $\mathbf{L\wedge L}$ cannot be null for otherwise
from \emph{Eq.(\ref{c16}) it would be necessary that }$m=0$ or $s=0.$
\end{remark}

Observe that\ the spin $1/2$ \ wave function needs to satisfy the fourth order
equation%
\begin{equation}
\left(  \frac{1}{64\ell^{2}}(\mathbf{L\wedge L})(\mathbf{L\wedge L}%
)-m^{2}s(s+1)\right)  \phi=0. \label{c23}%
\end{equation}
\ 

We can factorize$\ $the invariant $C_{2}$ as
\begin{equation}
\left(  \frac{1}{8\ell}\mathbf{L\wedge L}+m\sqrt{s(s+1)}\right)  \left(
\frac{1}{8\ell}\mathbf{L\wedge L}-m\sqrt{s(s+1)}\right)  =0. \label{c18}%
\end{equation}

Then, a possible second order equation that we will impose to be satisfied by
$\phi$ is\footnote{We used that $s=1/2$.}%
\begin{equation}
\left(  \mathbf{L\wedge L}-4\sqrt{3}m\ell\right)  \phi=0. \label{c18a}%
\end{equation}

To continue observe that we cannot factorize $\mathbf{L\lrcorner L}-\ell
^{2}m^{2}=0$ in two first order operators. However taking into account
Eq.(\ref{c9}) we can write
\begin{equation}
\frac{1}{\ell^{2}}\mathbf{L}^{2}-m^{2}-\frac{1}{\ell^{2}}\mathbf{L\wedge L}=0.
\label{c19}%
\end{equation}
Thus, taking into account Eq.(\ref{c18a}) we can write
\begin{equation}
\left(  \frac{1}{\ell^{2}}\mathbf{L}^{2}-m^{2}\right)  \phi=\frac{1}{\ell^{2}%
}(\mathbf{L\wedge L)}\phi=\frac{4\sqrt{3}m}{\ell}\phi, \label{c199}%
\end{equation}
or
\begin{equation}
\left(  \frac{1}{\ell^{2}}\mathbf{L}^{2}-m^{2}-\frac{4\sqrt{3}m}{\ell}\right)
\phi=0 \label{c19a}%
\end{equation}
and calling
\begin{equation}
\boldsymbol{\lambda}^{2}:=m^{2}+\frac{4\sqrt{3}m}{\ell} \label{c19aa}%
\end{equation}
we can now factor Eq.(\ref{c19a}) as
\begin{equation}
\left(  \frac{1}{\ell^{2}}\mathbf{L}^{2}-\boldsymbol{\lambda}^{2}\right)
\phi=\left(  \frac{1}{\ell}\mathbf{L+}\boldsymbol{\lambda}\right)  \left(
\frac{1}{\ell}\mathbf{L}-\boldsymbol{\lambda}\right)  \phi=0. \label{c21}%
\end{equation}

\subsection{The Tangency Condition and the DHESS1}

Let $\{\boldsymbol{\mathring{\theta}}^{0},\boldsymbol{\mathring{\theta}}%
^{1},\boldsymbol{\mathring{\theta}}^{2},\boldsymbol{\mathring{\theta}}%
^{3},\boldsymbol{\mathring{\theta}}^{4}\}$ be an orthonormal basis for $%
{\textstyle\bigwedge\nolimits^{1}}
T^{\ast}\mathring{M}$ such that $\{\boldsymbol{\theta}^{\mu}%
=\boldsymbol{\mathring{\theta}}^{\mu},~~~\mu=0,1,2,3\}$ is a tangent cotetrad
basis for de Sitter spacetime, i.e., $\boldsymbol{\theta}^{\mu}$ $\in\sec%
{\textstyle\bigwedge\nolimits^{1}}
T^{\ast}M\hookrightarrow\sec\mathcal{C\ell}(M,\boldsymbol{g})\subset
\sec\mathcal{C\ell}(\mathring{M},\boldsymbol{\mathring{g}})$ with
$\boldsymbol{\mathring{\theta}}^{4}$ orthogonal to $M$, i.e.,
$\boldsymbol{\mathring{\theta}}^{4}\lrcorner\tau_{\boldsymbol{g}}=0$.

We now propose taking into account Eq.(\ref{c21}) that the electron wave
function in de Sitter spacetime must satisfy the linear equation%

\begin{equation}
\left(  \frac{1}{\ell}\mathbf{L}-\boldsymbol{\lambda}\right)  \phi=0,
\label{c22}%
\end{equation}
with the constrain that $\phi$ is tangent to $M$, i.e., it does not contain in
its expansion terms containing $\boldsymbol{\mathring{\theta}}^{A}%
\boldsymbol{\mathring{\theta}}^{4}$.

Eq.(\ref{c22}) will be called the Dirac-Hestenes equation in de Sitter
spacetime (\textbf{DHESS1}).

\begin{remark}
\label{lambda}Take notice that in our formalism $\boldsymbol{\lambda}^{2}$
must be necessarily be a real number. Thus, if we had choose as wave equation
from the factorization of \emph{Eq.(\ref{c23})} the equation%
\begin{equation}
\left(  \mathbf{L\wedge L}+4\sqrt{3}m\ell\right)  \phi=0.
\end{equation}
we would get $\boldsymbol{\lambda}^{2}=m(m-\frac{4\sqrt{3}}{\ell})\geq0$ and
would arrive at the conclusion that the theory implies that all particles
living in de Sitter manifold would have masses satisfying $m\geq\frac
{4\sqrt{3}}{l}$!. This will be investigate in another publication.
\end{remark}

\begin{remark}
We make the important observation that if we did not use \emph{Eq.(\ref{c18a}%
)}, a constraint coming for the factorization of the second Casimir invariant
$C_{2}$, then factorization of \emph{Eq.(\ref{c19})} leads \emph{(}taking into
account that $\star^{-1}\mathbf{W=L\wedge L}$\emph{)} to the
integro-differential equation%
\begin{equation}
\left(  \frac{1}{\ell}\mathbf{L}-\sqrt{m^{2}+\star^{-1}\frac{1}{\ell^{2}%
}\mathbf{W}}\right)  \phi=0. \label{c240}%
\end{equation}
\ 
\end{remark}

\begin{remark}
Note that the matrix representation of \emph{Eq.(\ref{c21}) }in terms of the
representatives $\Phi\in\sec\mathbf{P}_{\mathrm{Spin}_{4,1}^{e}}(\mathring
{M},\mathtt{\mathring{g})}\times\mathbb{C}^{4}$ of $\phi$ is
\begin{equation}
\left(  \frac{1}{\ell}\underline{L}-\boldsymbol{\lambda}\right)
\Phi=0,~~~~\underline{L}:=\underline{\gamma}^{A}\underline{\gamma}%
^{B}\mathbf{L}_{AB} \label{c25}%
\end{equation}
In \emph{Eq.(\ref{c25})} $\underline{\gamma}^{A}$ is a matrix representation
of $\theta^{A}$, $A=1,2,3,4,0$. Note that since the operator $\underline{L}$
is not Hermitian in \emph{\cite{gursey}} the author left open the possibility
that $\boldsymbol{\lambda}$\ is a complex number. This is not the case in our
formalism, as we already observed in Remark \emph{\ref{lambda}}.
\end{remark}

\section{An Heuristic Derivation of the DHESS2}

We start recalling that a classical free particle in de Sitter manifold
structure $(M,\boldsymbol{g})$ certainly follows a timelike worldline
$\sigma:\mathbb{R}\supset I\rightarrow M$. To unveil the nature of that motion
we suppose the existence of a 2-form field $\boldsymbol{L}\in\sec%
{\textstyle\bigwedge\nolimits^{2}}
T^{\ast}\mathring{M}\hookrightarrow\sec\mathcal{C\ell}((\mathring
{M},\boldsymbol{\mathring{g}})$ such that its restriction over $\sigma$ is
$\boldsymbol{l}$ i.e., $\boldsymbol{L}_{\mid\sigma}=\boldsymbol{l}$. given by
Eq.(\ref{c1}).

It is a very reasonable hypothesis that the classical motion of a free
particle in de Sitter manifold structure $(M,\boldsymbol{g})$ will happen only
under the condition that $\boldsymbol{L}\in\sec%
{\textstyle\bigwedge\nolimits^{2}}
T^{\ast}\mathring{M}\hookrightarrow\sec\mathcal{C\ell}((\mathring
{M},\boldsymbol{\mathring{g}})$ is a constant $2$-form $\boldsymbol{B}$\ as
registered by an hypothetical \textquotedblleft observer\textquotedblright%
\ living in the bulk spacetime structure $\mathbb{R}^{4,1}$. This is an
obvious generalization of the fact that the free motion of a classical
particle in the Minkowski spacetime structure (see Appendix D) happens with
constant momentum. Let $\{\boldsymbol{\mathring{\theta}}^{A},~A=0,1,2,3,4\}$
be an orthonormal basis for $%
{\textstyle\bigwedge\nolimits^{1}}
T^{\ast}\mathring{M}$ such that $\{\boldsymbol{\theta}^{\mu}%
=\boldsymbol{\mathring{\theta}}^{\mu},~\mu=0,1,2,3\}$ is a cotetrad basis for
$%
{\textstyle\bigwedge\nolimits^{1}}
T^{\ast}M\hookrightarrow\mathcal{C\ell}(M,\mathtt{g})\subset\mathcal{C\ell
}(\mathring{M},\mathtt{\mathring{g}})$ with $\boldsymbol{\mathring{\theta}%
}^{4}$ orthogonal to $M$, i.e., $\boldsymbol{\mathring{\theta}}^{4}%
\lrcorner\tau_{\boldsymbol{g}}=0$. Then, the condition that $\boldsymbol{L}$
is a constant $2$-form may be written
\begin{equation}
\frac{1}{\ell}\boldsymbol{L}=-\kappa\boldsymbol{B}. \label{C32}%
\end{equation}

Now, let $\phi^{\prime}$ be an invertible \textbf{DHSF }living in
$(\mathring{M},\boldsymbol{\mathring{g})}$\textbf{ }such that
\begin{equation}
\boldsymbol{B}=\phi^{\prime-1}\boldsymbol{\mathring{\theta}}^{1}%
\boldsymbol{\mathring{\theta}}^{2}\phi^{\prime}=\boldsymbol{\mathring{E}}%
^{1}\boldsymbol{\mathring{E}}^{2}. \label{c33}%
\end{equation}

Before continuing we also suppose that
\begin{equation}
\phi^{\prime-1}\boldsymbol{\mathring{\theta}}^{A}\boldsymbol{\mathring{\theta
}}^{B}\phi^{\prime}=\boldsymbol{\mathring{E}}^{A}\boldsymbol{\mathring{E}}%
^{B},~~~A,B=0,1,2,3,4. \label{c33a}%
\end{equation}

Using Eq.(\ref{c33}) in Eq.(\ref{C32}) we can write a purely classic
\textbf{DHESS} equation, namely,%
\begin{equation}
\frac{1}{\ell}\boldsymbol{L}\phi^{\prime-1}\boldsymbol{\mathring{\theta}}%
^{2}\boldsymbol{\mathring{\theta}}^{1}=-\kappa\phi^{\prime-1}. \label{C34}%
\end{equation}
This is compatible with a quantum \textbf{DHESS1}, i.e.,
\begin{equation}
\frac{1}{\ell}\mathbf{L}\phi^{\prime-1}+\kappa\phi^{\prime-1}=0 \label{c13}%
\end{equation}
with the \emph{postulate} that when we restrict our considerations to
\textbf{DHSFs} living in the de Sitter structure $(M,\boldsymbol{g})$ it is:%

\begin{equation}
\frac{1}{\ell}\boldsymbol{L}\phi^{\prime-1}\boldsymbol{\mathring{\theta}}%
^{2}\boldsymbol{\mathring{\theta}}^{1}\mapsto\frac{1}{\ell}\mathbf{L}%
\phi^{\prime-1}:=\boldsymbol{\mathring{\theta}}^{A}\boldsymbol{\mathring
{\theta}}^{B}\left(  X_{A}\mathbf{P}_{B}-X_{B}\mathbf{P}_{A}\right)
\phi^{\prime-1}\label{c35}%
\end{equation}
with
\begin{equation}
\boldsymbol{P}_{A}\phi^{\prime-1}\mapsto\mathbf{P}_{A}\phi^{\prime-1}%
:=\frac{\partial}{\partial X^{A}}\phi^{\prime-1}\boldsymbol{\mathring{\theta}%
}^{2}\boldsymbol{\mathring{\theta}}^{1}.\label{c36}%
\end{equation}
Note that under the above conditions the \textbf{DHESS1} equation
(Eq.(\ref{c22}))\ will be identical to the \textbf{DHESS2 }(Eq.(\ref{c13}))
if
\begin{equation}
\lambda\mapsto\kappa,~~~~\phi^{\prime-1}\mapsto\phi.\label{C366}%
\end{equation}
Indeed, recall that a generalization of a \textbf{DHSF} (see Eq.(\ref{ddhsf})
in Appendix B ) for the $(\mathring{M}\simeq\mathbb{R}^{5}%
,\boldsymbol{\mathring{g}})$ structure (see page 106 of\cite{hs1984}) can be
written as
\begin{equation}
\phi^{\prime}=\rho^{1/2}U,~~~\phi^{\prime}=\rho^{-1/2}U^{-1}%
\end{equation}
where $\rho$ is a scalar function $U\in\sec\mathcal{C\ell}^{0}\mathcal{(}%
\mathring{M},\boldsymbol{\mathring{g})}$ and for any $x\in\mathring{M}$,
$U(x)\tilde{U}(x)=1$. Note that $U$ can be written as%
\begin{equation}
U(x)=e^{z(x)F(x)}=(\cosh z(x)+F(x)\sinh z(x))
\end{equation}
with the $2$-form $F$ satisfying $F^{2}=1$. Under these conditions a standard
calculation shows that%

\begin{equation}
\frac{\partial\phi^{\prime^{-1}}}{\partial X^{M}}\phi^{\prime}=-\phi
^{\prime^{-1}}\frac{\partial\phi^{\prime}}{\partial X^{M}}=\phi^{\prime}%
\frac{\partial\phi^{\prime-1}}{\partial X^{M}},~~M=0,1,2,3,4. \label{cx}%
\end{equation}
Thus, starting with Eq.(\ref{c22}) and writing $\boldsymbol{\mathring{\theta}%
}^{A}\boldsymbol{\mathring{\theta}}^{B}=\phi^{\prime}\boldsymbol{\mathring{E}%
}^{A}\boldsymbol{\mathring{E}}^{B}\phi^{\prime-1}$we get%

\begin{align}
0  &  =\boldsymbol{\mathring{\theta}}^{A}\boldsymbol{\mathring{\theta}}%
^{B}\left(  X_{A}\boldsymbol{P}_{B}-X_{B}\boldsymbol{P}_{A}\right)
\phi^{\prime-1}+\kappa\phi^{\prime-1}\nonumber\\
&  =-\phi^{\prime}\boldsymbol{\mathring{E}}^{A}\boldsymbol{\mathring{E}}%
^{B}\phi^{\prime-1}\left[  \left(  X_{A}\frac{\partial}{\partial X^{B}}%
-X_{B}\frac{\partial}{\partial X^{A}}\right)  \phi^{\prime-1}\right]
\phi^{\prime}\boldsymbol{\mathring{E}}^{2}\boldsymbol{\mathring{E}}^{1}%
\phi^{\prime-1}+\kappa\phi^{\prime-1}\nonumber\\
&  =-\phi^{\prime}\boldsymbol{\mathring{E}}^{A}\boldsymbol{\mathring{E}}%
^{B}\phi^{\prime-1}\left[  \phi^{\prime}\left(  X_{A}\frac{\partial}{\partial
X^{B}}-X_{B}\frac{\partial}{\partial X^{A}}\right)  \phi^{\prime-1}\right]
\boldsymbol{\mathring{E}}^{2}\boldsymbol{\mathring{E}}^{1}\phi^{\prime
-1}+\kappa\phi^{\prime-1}\nonumber\\
&  =-\phi^{\prime}\boldsymbol{\mathring{E}}^{A}\boldsymbol{\mathring{E}}%
^{B}\left[  \left(  X_{A}\frac{\partial}{\partial X^{B}}-X_{B}\frac{\partial
}{\partial X^{A}}\right)  \phi^{\prime-1}\right]  \boldsymbol{\mathring{E}%
}^{2}\boldsymbol{\mathring{E}}^{1}\phi^{\prime^{-1}}+\kappa\phi^{\prime-1}
\label{C36A}%
\end{align}
and multiplying Eq.(\ref{C36A}) on the left by $\phi^{\prime-1}$ and on the
right by $\phi$ we get%
\begin{equation}
0=\boldsymbol{\mathring{E}}^{A}\boldsymbol{\mathring{E}}^{B}\left[  \left(
X_{A}\frac{\partial}{\partial X^{B}}-X_{B}\frac{\partial}{\partial X^{A}%
}\right)  \phi^{\prime-1}\right]  \boldsymbol{\mathring{E}}^{2}%
\boldsymbol{\mathring{E}}^{1}-\kappa\phi^{\prime-1}. \label{c36b}%
\end{equation}
which taking into account Eq.(\ref{C366}) is identical to Eq.(\ref{c22}).

\section{The limit $\ell\rightarrow\infty$ of Eq.(\ref{c13})}

Expressing $\mathbf{L}_{AB}$ in terms of the projective coordinates we get%
\begin{equation}
\mathbf{L}_{\alpha4}\phi=l\partial_{\alpha}\phi\boldsymbol{\mathring{\theta}%
}^{2}\boldsymbol{\mathring{\theta}}^{1}-\frac{1}{4l}(2\eta_{\alpha\lambda
}x^{\lambda}x^{\nu}-\sigma^{2}\delta_{\alpha}^{\nu})\partial_{\nu}%
\phi\boldsymbol{\mathring{\theta}}^{2}\boldsymbol{\mathring{\theta}}%
^{1}.\label{c38}%
\end{equation}
Also\footnote{In written Eq.(\ref{c39}) we take into account that the metric
of the de Sitter manifold has signature $(1,-1,-1,-1)$ and the metric of the
bulk manifold has signature $(1,1,1,1,-1)$.},
\begin{equation}
\mathbf{L}_{\mu\nu}\phi=-\eta_{\mu\lambda}x^{\lambda}\partial_{\nu}%
\phi\boldsymbol{\mathring{\theta}}^{2}\boldsymbol{\mathring{\theta}}^{1}%
+\eta_{\nu\lambda}x^{\lambda}\partial_{\mu}\phi\boldsymbol{\mathring{\theta}%
}^{2}\boldsymbol{\mathring{\theta}}^{1}\label{c39}%
\end{equation}
Taking into account the results of Appendix B (Eq.(\ref{419}) and Remark
\ref{XX} ) we know that $\phi$ must be of the form
\begin{equation}
\phi=\left(  \varphi+\boldsymbol{\mathring{\theta}}_{4}\boldsymbol{\mathring
{\theta}}_{\alpha}\xi^{\alpha}\right)  \in\sec\mathcal{C\ell}^{0}%
\mathcal{(}\mathring{M},\boldsymbol{\mathring{g}}\mathcal{)}\label{c39aa}%
\end{equation}
with $\varphi$ living in a particular ideal of $\mathcal{C\ell}^{0}%
\mathcal{(}\mathring{M},\boldsymbol{\mathring{g}}\mathcal{)}$. Moreover,
taking into account that in Eq.(\ref{419}) the $\mathbf{\xi}^{\alpha}$,
$\alpha=0,1,2,3$ must be linked in order to define only three degrees of
freedom we take $\eta_{\alpha\beta}\xi^{\alpha}\tilde{\xi}^{\beta}=1$.
Recalling that $\eta_{\alpha\beta}X^{\alpha}X^{\beta}=\ell^{2}(1-(X^{4}%
)^{2}/\ell^{2})$ we have $\eta_{\alpha\beta}x^{\alpha}x^{\beta}=-\frac
{\ell^{2}}{\Omega^{2}}(1-\frac{(X^{4})^{2}}{\ell^{2}})$ and a possible
solution for our constraint is%
\begin{equation}
\xi^{\alpha}=\frac{1}{\ell}\Omega^{2}x^{\alpha}\lambda\label{c39x}%
\end{equation}
where $\lambda$ $\in\mathcal{C\ell}^{0}\mathcal{(}\mathring{M}%
,\boldsymbol{\mathring{g}}\mathcal{)}$ is such that $\lambda\tilde{\lambda
}=-1/(1-(X^{4})^{2}/\ell^{2})$. Thus $\phi$ can be written as
\begin{equation}
\phi=\varphi+\frac{1}{\ell}\Omega^{2}\boldsymbol{\mathring{\theta}}%
_{4}\boldsymbol{\mathring{\theta}}_{\alpha}x^{\alpha}\lambda\label{c39a}%
\end{equation}
and Eq.(\ref{c35}) can be written (in the limit $\ell\rightarrow\infty$) as
\begin{equation}
\boldsymbol{\mathring{\theta}}^{\alpha}\boldsymbol{\mathring{\theta}}^{4}%
\frac{\partial}{\partial x^{\alpha}}\varphi\boldsymbol{\mathring{\theta}}%
^{2}\boldsymbol{\mathring{\theta}}^{1}-\boldsymbol{\lambda}\varphi
=0.\label{c400}%
\end{equation}

Now, calling $\boldsymbol{\mathring{\theta}}^{\alpha}\boldsymbol{\mathring
{\theta}}^{4}=\Gamma^{\alpha}$, $\alpha=0,1,2,3$ we know (see AppendixD) that
\begin{equation}
\Gamma^{\alpha}\Gamma^{\beta}+\Gamma^{\beta}\Gamma^{\alpha}=2\eta^{\alpha
\beta}.\label{c41}%
\end{equation}

Finally, recalling Eq.(\ref{c39a}) we can write Eq.(\ref{c400}) in the limit
$\ell\rightarrow\infty$ as
\begin{equation}
\Gamma^{\alpha}\frac{\partial}{\partial x^{\alpha}}\varphi\Gamma^{2}\Gamma
^{1}-m\varphi=0.\label{C43}%
\end{equation}
which is clearly a \emph{representative} of the \textbf{DHE} in Minkowski
spacetime in the $\mathcal{C\ell(}M\simeq\mathbb{R}^{4},\boldsymbol{\eta
}\mathcal{)}$ bundle which reads\footnote{In Eq.(\ref{c44}) \ $\{x^{\mu}\}$
are coordinates in Einstein-Lorentz-Poincar\'{e} gauge, and $\gamma^{\mu
}:=dx^{\mu}$. More details, in Appendix C.}%
\begin{equation}
\gamma^{\alpha}\frac{\partial}{\partial x^{\alpha}}\psi\gamma^{2}\gamma
^{1}-m\psi\gamma^{0}=0.\label{c44}%
\end{equation}
Indeed, multiplying Eq.(\ref{c44}) on the right by the idempotent $\frac{1}%
{2}(1+\gamma^{0})$ it reads (calling $\zeta=\psi\frac{1}{2}(1+\gamma^{0})$)%
\begin{equation}
\gamma^{\alpha}\frac{\partial}{\partial x^{\alpha}}\zeta\gamma^{2}\gamma
^{1}-m\zeta=0.\label{c45}%
\end{equation}

So, Eqs. (\ref{C43}) and (\ref{c44}) can be identified with the
identifications
\begin{equation}
\Gamma^{\alpha}\leftrightarrow\gamma^{\alpha},~~~~~\varphi\leftrightarrow
\zeta\label{c45a}%
\end{equation}

\begin{remark}
It is well known that when $\ell\rightarrow\infty$ the $\mathbf{DHESS1}%
$\textbf{ }\emph{(Eq.(\ref{c22})) \ which is a Clifford bundle representation
of the Dirac equation (written with the standard matrix formalism) is also
equivalent to the }$\mathbf{DHE}$ in Minkowski spacetime \cite{gursey}.
\end{remark}

\section{Conclusions}

We gave a Clifford bundle motivated approach to the wave equation of a free
spin $1/2$ fermion in the de Sitter manifold, a brane with topology
$M=\mathrm{S0}(4,1)/\mathrm{S0}(3,1)$ living in the bulk spacetime
$\mathbb{R}^{4,1}=(\mathring{M}=\mathbb{R}^{5},\boldsymbol{\mathring{g}})$ and
equipped with a metric field $\boldsymbol{g:}=-i^{\ast}\mathring{g}$ with
$\boldsymbol{i}:M\rightarrow\mathring{M}$ being the inclusion map. To obtain
the analog of Dirac equation in Minkowski spacetime we appropriately factorize
the two Casimir invariants $C_{1}$ and $C_{2}$ of the Lie algebra of the de
Sitter group using the constraint given in the linearization of $C_{2}$ as
input to linearize $C_{1}$. In this way we obtain an equation that we called
\textbf{DHESS1 }(which is simply postulated in previous studies
(\cite{dirac,gursey}). Next we derive a wave equation (called \textbf{DHESS2})
for a free spin $1/2$ fermion in the de Sitter manifold using an heuristic
argument which is an obvious generalization of an heuristic argument
(described in detail in one of the appendices) permitting a derivation of the
Dirac equation in Minkowski spacetime which shows that such famous equation
express nothing more that the momentum of a free particle is a constant vector
field over timelike \ integral curves of a given velocity field. \ It is a
nice fact that \textbf{DHESS1} and \textbf{DHESS2 }coincide. We emphasize
moreover that our approach leaves clear the \emph{nature and meaning of the
Casimir invariants} \cite{cota} and thus of the object $\boldsymbol{\lambda}$
(Eq.(\ref{c19aa})), something that is not clear in other papers on the subject
such as, e.g., \cite{dirac,gursey,NC1999,riordan} which use the standard
covariant Dirac spinor fields.

As a last comment here we recall that if the de Sitter manifold is supposed to
be a spacetime, i.e., a structure\cite{rc2007} $(M,\boldsymbol{g}$%
,$\tau_{\boldsymbol{g}},\nabla,\uparrow)$ where $\nabla$ is an arbitrary
connection compatible with $\boldsymbol{g}$ then the writing of Dirac equation
in such a structure is supposed to \ be given by very different arguments from
the ones used in this paper. A comparison of the two approaches will be
presented elsewhere.

\begin{acknowledgement}
M. Rivera-Tapia and E. A. Notte-Cuello were supported by the Direccion de
Investigaci\'{o}n de la Universidad de La Serena DIULS. The work of I.
Kondrashuk was supported by Fondecyt (Chile) Grants Nos. 1040368, 1050512, and
1121030, and by DIUBB (Chile) Grants Nos. 102609 and 121909 GI/C-UBB. S. A.
Wainer is grateful to CAPES for a Ph.D. fellowship.
\end{acknowledgement}

\appendix

\section{\textrm{SO}$^{e}(4,1)$ and \textrm{Spin}$_{4,1}^{e}$ and their Lie
Algebras}

The group\textrm{ O}$(4,1)$ may be defined as the group of (invertible)
$5\times5$ real matrices $\mathbf{\Lambda}$ such that if $\mathbf{X}%
^{t}=(X^{1},X^{2},X^{3},X^{4},X^{0})$ denotes a $1\times5$ real matrix and if
\begin{equation}
\boldsymbol{\mathring{G}=}\left(
\begin{array}
[c]{ccccc}%
1 & 0 & 0 & 0 & 0\\
0 & 1 & 0 & 0 & 0\\
0 & 0 & 1 & 0 & 0\\
0 & 0 & 0 & 1 & 0\\
0 & 0 & 0 & 0 & -1
\end{array}
\right)  \label{A1}%
\end{equation}
then%
\begin{equation}
\mathbf{\Lambda}^{t}\boldsymbol{\mathring{G}}\mathbf{\Lambda}%
=\boldsymbol{\mathring{G}.} \label{a2}%
\end{equation}

Of, course, $\det\mathbf{\Lambda}^{2}=1$. Let \textrm{O}$^{\uparrow}(4,1)$ be
the subgroup of \textrm{O}$(4,1)$ such that if $\mathbf{X}^{t}=(X^{1}%
,X^{2},X^{3},X^{4},X^{0})$ and $X^{0}>0$ then $(\mathbf{\Lambda X)}%
^{t}=(X^{\prime1},X^{\prime2},X^{\prime3},X^{^{\prime}4},X^{\prime0})$ has
$X^{\prime0}>0.$

The matrices with $\det\Lambda=1$ closes the subgroup \textrm{SO}$_{+}(4,1)$
of \textrm{O}$(4,1)$. The elements of \textrm{SO}$_{+}(4,1)$ are clearly
connected to the identity element of \textrm{O}$(4,1)$. We shall denoted by
\textrm{SO}$^{e}(4,1)=$ \textrm{O}$^{\uparrow}(4,1)\cap$\textrm{SO}$_{+}%
(4,1)$. We denote by \textrm{Spin}$_{4,1}^{e}$the simple connected group that
is the double cover of \textrm{SO}$^{e}(4,1)$. It is a well known result
\cite{porteous} that the elements of \textrm{Spin}$_{4,1}^{e}$ are the
invertible elements of $u\in\mathbb{R}_{4,1}^{0}$ such that $u\tilde{u}=1$.

Now, it is a well known result that the elements of \textrm{SO}$^{e}(4,1)$ can
be written as an exponential of a sum of antisymmetric matrices $\mathbf{M}%
_{AB}$, i.e.,
\begin{equation}
\mathbf{\Lambda=}\exp\left(  \frac{1}{2}\chi^{AB}\mathbf{M}_{AB}\right)
,~~~~~\chi^{AB}=-\chi^{BA},~~~~\chi^{AB}\in\mathbb{R}\text{.} \label{a3}%
\end{equation}
The matrices $\mathbf{M}_{AB}$ closes the Lie algebra \textrm{so}$^{e}(4,1)$
of \textrm{SO}$_{+}(4,1)$ and satisfy the commutation relations%
\begin{equation}
\lbrack\mathbf{M}_{AB},\mathbf{M}_{CD}]=\eta_{AC}\mathbf{M}_{BD}+\eta
_{BD}\mathbf{M}_{AC}-\eta_{BC}\mathbf{M}_{AD}-\eta_{AD}\mathbf{M}_{BC}.
\label{a4}%
\end{equation}

When \textrm{SO}$_{+}(4,1)$ acts as a transformation group in the manifold
$\mathring{M}$ the generators $\mathbf{M}_{AB}$ of the Lie algebra
\textrm{so}$^{e}(4,1)$ are represented by the vector fields
\begin{equation}
\mathbf{\xi}_{AB}=\eta_{AC}X^{C}\frac{\partial}{\partial X^{B}}-\eta_{BC}%
X^{C}\frac{\partial}{\partial X^{A}} \label{a5}%
\end{equation}
which are Killing vector fields of $\boldsymbol{\mathring{g}}$, i.e.,
$\pounds _{\mathcal{\xi}_{AB}}\boldsymbol{\mathring{g}}=0$, with $\pounds $
denoting the Lie derivative. Of course, if $f:\mathring{M}\rightarrow
\mathbb{R}$ we immediately verify that%
\begin{equation}
\lbrack\mathbf{\xi}_{AB},\mathbf{\xi}_{CD}]f=(\eta_{AC}\mathbf{\xi}_{BD}%
+\eta_{BD}\mathbf{\xi}_{AC}-\eta_{BC}\mathbf{\xi}_{AD}-\eta_{AD}\mathbf{\xi
}_{BC})f. \label{a6}%
\end{equation}

On the other hand the elements of \textrm{Spin}$_{4,1}^{e}$ are of the form
$\pm\exp\mathbf{B}$ where $\mathbf{B}$ is a biform\footnote{See page 223 of
\cite{lounesto}.}. We write $u\in$\textrm{Spin}$_{4,1}^{e}$ as%
\begin{equation}
u=\exp\left(  \frac{1}{4}\chi_{AB}\boldsymbol{\mathring{E}}^{AB}\right)
=\exp\left(  \frac{1}{4}\chi_{AB}\boldsymbol{\mathring{E}}^{A}%
\boldsymbol{\mathring{E}}^{B}\right)  \label{A7}%
\end{equation}
and we may immediately verify that the biforms $\mathbf{S}^{AB}=\frac{1}%
{2}\boldsymbol{\mathring{E}}^{AB}$ satisfy the Lie algebra \textrm{spin}%
$_{4,1}^{e}$which is isomorphic to the Lie algebra \textrm{so}$^{e}(4,1)$, i.e.%

\begin{equation}
\lbrack\mathbf{S}_{AB},\mathbf{S}_{CD}]=\eta_{AC}\mathbf{S}_{BD}+\eta
_{BD}\mathbf{S}_{AC}-\eta_{BC}\mathbf{S}_{AD}-\eta_{AD}\mathbf{S}_{BC}.
\label{a8}%
\end{equation}

\section{DHSF}

In this Appendix for convenience for the reader we recall the definition of
the concept of \textbf{DHSF} for a\ generic $4$-dimensional spin manifold $M$
equipped with a metric $\boldsymbol{g}$. The presentation improves a lit bit
the theory as developed originally in \cite{r2004,mr02004,lrw}.\footnote{Also
in \cite{lrw} it is presented a geometrical inspired theory to the Lie
derivative of \textbf{DHSFs}.}

In what follows $\mathbf{P}_{\mathrm{SO}_{1,3}^{e}}(M,\boldsymbol{g})$
($P_{\mathrm{SO}_{1,3}^{e}}(M,\mathtt{g})$) denotes the principal bundle of
oriented Lorentz tetrads (cotetrads).

\begin{definition}
A spin structure for a general $m$-dimensional manifold $M$ consists of a
principal fiber bundle $\pi_{s}:P_{\mathrm{Spin}_{p,q}^{e}}(M,\mathtt{g}%
)\rightarrow M$, \emph{(}called the Spin Frame Bundle\emph{)} with group
$\mathrm{Spin}_{p,q}^{e}$ and a map%
\begin{equation}
\Lambda:P_{\mathrm{Spin}_{p,q}^{e}}(M,\mathtt{g})\rightarrow P_{\mathrm{SO}%
_{p,q}^{e}}(M,\mathtt{g}),
\end{equation}
satisfying the following conditions:

\begin{description}
\item[(i)] $\pi(\Lambda(p))=\pi_{s}(p),\forall p\in P_{\mathrm{Spin}_{p,q}%
^{e}}(M,\mathtt{g})$, where $\pi$ is the projection map of the bundle
$\pi:P_{\mathrm{SO}_{p,q}^{e}}(M,\mathtt{g})\rightarrow M$.

\item[(ii)] $\Lambda(pu)=\Lambda(p)\mathrm{Ad}_{u},\forall p\in
P_{\mathrm{Spin}_{p,q}^{e}}(M,\mathtt{g})$ and $\mathrm{Ad}:\mathrm{Spin}%
_{p,q}^{e}\rightarrow\mathrm{SO}_{p,q}^{e}, \break\mathrm{Ad}_{u}%
(a)=uau^{-1}.$
\end{description}
\end{definition}

\begin{definition}
Any section of $P_{\mathrm{Spin}_{p,q}^{e}}(M,\mathtt{g})$ is called a spin
frame field \emph{(}or simply a spin frame\emph{)}. We shall use the symbol
$\Xi\in\sec P_{\mathrm{Spin}_{p,q}^{e}}(M,\mathtt{g})$ to denoted a spin frame.
\end{definition}

We know that\footnote{Where $Ad:\mathrm{Spin}_{1,3}^{e}\rightarrow
\mathrm{End(}\mathbb{R}_{1,3}\mathrm{)}$\textrm{ } is such that
$Ad(u)a=uau^{-1}$. And $\rho:\mathrm{SO}_{1,3}^{e}\rightarrow\mathrm{End(}%
\mathbb{R}_{1,3}\mathrm{)}$ is the natural action of $\mathrm{SO}_{1,3}^{e}$
on $\mathbb{R}_{1,3}$.} \cite{rc2007}:%
\begin{equation}
\mathcal{C}\ell(M,\mathtt{g})=P_{\mathrm{SO}_{1,3}^{e}}(M,\mathtt{g}%
)\times_{\rho}\mathbb{R}_{1,3}=P_{\mathrm{Spin}_{1,3}^{e}}(M,\mathtt{g}%
)\times_{Ad}\mathbb{R}_{1,3}, \label{eq_cliff}%
\end{equation}
and since\footnote{Given the objets $A$ and $B$, $A$ $\hookrightarrow$ $B$
means as usual that $A$ is embedded in $B$ and moreover, $A\subseteq B$. In
particular, recall that there is a canonical vector space isomorphism between
$\bigwedge\mathbb{R}^{1,3}$ and $\mathbb{R}_{1,3}$, which is written
$\bigwedge\mathbb{R}^{1,3}\hookrightarrow\mathbb{R}_{1,3}$. Details in
\cite{cru,lawmi}.} $\bigwedge TM\hookrightarrow\mathcal{C}\ell(M,\mathtt{g})$,
sections of $\mathcal{C}\ell(M,\mathtt{g})$ (the Clifford fields) can be
represented as a sum of non homogeneous differential forms.

Next, using that $M\simeq S^{3}\times\mathbb{R\subset}\mathring{M}$ is
parallelizable\footnote{Follows by the fact that $S^{3}$ is a Lie group}, we
introduce the global tetrad basis $\boldsymbol{e}_{\alpha},\alpha=0,1,2,3$ on
$TM$ and in $T^{\ast}M$ the cotetrad basis on $\{\boldsymbol{\gamma}^{\alpha
}\}$, which are dual basis. We introduce the reciprocal basis
$\{\boldsymbol{e}^{\alpha}\}$ and $\{\boldsymbol{\gamma}_{\alpha}\}$ of
$\{\boldsymbol{e}_{\alpha}\}$ and $\{\boldsymbol{\gamma}^{\alpha}\}$
satisfying
\begin{equation}
\boldsymbol{g}(\boldsymbol{e}_{\alpha},\boldsymbol{e}^{\beta})=\delta_{\alpha
}^{\beta},~~~\mathtt{g}(\boldsymbol{\gamma}^{\beta},\boldsymbol{\gamma
}_{\alpha})=\delta_{\alpha}^{\beta}. \label{p2}%
\end{equation}

Moreover, recall that\footnote{Where the matrix with entries $\eta
_{\alpha\beta}$ (or $\eta^{\alpha\beta}$) is the diagonal matrix
$(1,-1,-1,-1)$.}%
\begin{equation}
\boldsymbol{g}=\eta_{\alpha\beta}\boldsymbol{\gamma}^{\alpha}\otimes
\boldsymbol{\gamma}^{\beta}=\eta^{\alpha\beta}\boldsymbol{\gamma}_{\alpha
}\otimes\boldsymbol{\gamma}_{\beta},~~~\mathtt{g}=\eta^{\alpha\beta
}\boldsymbol{e}_{\alpha}\otimes\boldsymbol{e}_{\beta}=\eta_{\alpha\beta
}\boldsymbol{e}^{\alpha}\otimes\boldsymbol{e}^{\beta}. \label{P3}%
\end{equation}

In this work we have that exists a spin structure on\ the 4-dimensional
Lorentzian manifold $(M,\boldsymbol{g})$, since $M$ is parallelizable, i.e.,
$P_{\mathrm{SO}_{1,3}^{e}}(M,\mathtt{g})$ is trivial, because of the following
result due to Geroch \cite{geroch}:

\begin{theorem}
\label{gerochh}For a $4-$dimensional Lorentzian manifold$\ (M,\boldsymbol{g}%
)$, a spin structure exists if and only if $P_{\mathrm{SO}_{1,3}^{e}%
}(M,\mathtt{g})$ is a trivial bundle.
\end{theorem}

The basis $\left.  \boldsymbol{\gamma}^{\alpha}\right\vert _{p}$ of
$T_{p}M\simeq\mathbb{R}^{1,3},p\in M$, generates the algebra $\mathcal{C}%
\ell(T_{p}M,\mathtt{g})\simeq\mathbb{R}_{1,3}$. We have that \cite{rc2007}%
\[
\mathrm{e}=\frac{1}{2}(1+\boldsymbol{\gamma}^{0})\in\mathbb{R}_{1,3}%
\]
is a primitive idempotent of $\mathbb{R}_{1,3}$ and%
\[
\mathrm{f}=\frac{1}{2}(1+\boldsymbol{\gamma}^{0})\frac{1}{2}%
(1+i\boldsymbol{\gamma}^{2}\boldsymbol{\gamma}^{1})\in\mathbb{C\otimes
R}_{1,3}%
\]
is a primitive idempotent of $\mathbb{C\otimes R}_{1,3}$. Now, let
$I=\mathbb{R}_{1,3}\mathrm{e}$ and $I_{\mathbb{C}}=\mathbb{C\otimes R}%
_{1,3}\mathrm{f}$ be respectively the minimal left ideals of $\mathbb{R}%
_{1,3}$ and $\mathbb{C\otimes R}_{1,3}$ generated by $\mathrm{e}$ and
$\mathrm{f}$. Any $\phi\in I$ can be written as%
\[
\phi=\psi\mathrm{e}%
\]
with $\psi\in\mathbb{R}_{1,3}^{0}$. Analogously, any $\mathbf{\phi}\in
I_{\mathbb{C}}$ can be written as%
\[
\psi\mathrm{e}\frac{1}{2}(1+i\boldsymbol{\gamma}^{2}\boldsymbol{\gamma}^{1})
\]
with $\psi\in\mathbb{R}_{1,3}^{0}.$ Recall moreover that $\mathbb{C\otimes
R}_{1,3}\simeq\mathbb{R}_{4,1}\simeq\mathbb{C}(4)$. We can verify that%
\[
\left(
\begin{array}
[c]{cccc}%
1 & 0 & 0 & 0\\
0 & 0 & 0 & 0\\
0 & 0 & 0 & 0\\
0 & 0 & 0 & 0
\end{array}
\right)
\]
is a primitive idempotent of $\mathbb{C}(4)$ which is a matrix representation
of $\mathrm{f}$. In that way, there is a bijection between column spinors,
i.e., elements of $\mathbb{C}^{4}$ and the elements of $I_{\mathbb{C}}$.

Recalling that $\mathrm{Spin}_{1,3}^{e}\hookrightarrow\mathbb{R}_{1,3}^{0}$,
we give:

\begin{definition}
The left \emph{(}respectively right\emph{)} real spin-Clifford bundle of the
spin manifold $M$ is the vector bundle $\mathcal{C}\ell_{\mathrm{Spin}}%
^{l}(M,\mathtt{g})=P_{\mathrm{Spin}_{1,3}^{e}}(M,\mathtt{g})\times
_{l}\mathbb{R}_{1,3}$ \emph{(}respectively $\mathcal{C}\ell_{\mathrm{Spin}%
}^{r}(M,\mathtt{g})=P_{\mathrm{Spin}_{1,3}^{e}}(M,\mathtt{g})\times
_{r}\mathbb{R}_{1,3}$\emph{)} where $l$ is the representation of
$\mathrm{Spin}_{1,3}^{e}$ on $\mathbb{R}_{1,3}$ given by $l(a)x=ax$
\emph{(}respectively, where $r$ is the representation of $\mathrm{Spin}%
_{1,3}^{e}$ on $\mathbb{R}_{1,3}$ given by $r(a)x=xa^{-1}$\emph{)}. Sections
of $\mathcal{C}\ell_{\mathrm{Spin}}^{l}(M,\mathtt{g})$ are called left
spin-Clifford fields \emph{(}respectively right spin-Clifford fields\emph{)}.
\end{definition}

\begin{definition}
Let $\mathbf{e,f}\in\mathcal{C}\ell_{\mathrm{Spin}_{1,3}^{e}}^{l}%
(M,\mathtt{g})$ be a primitive global idempotents \footnote{We know that
global primitive idempotents exist because $M$ is parallelizable.
\[
\mathbf{e=\mathbf{[(\Xi}_{0},\frac{1}{2}(1+\boldsymbol{\gamma}^{0}%
)\mathbf{)]},f}=\mathbf{[(\Xi}_{0},\frac{1}{2}(1+\boldsymbol{\gamma}^{0}%
)\frac{1}{2}(1+i\boldsymbol{\gamma}^{2}\boldsymbol{\gamma}^{1})\mathbf{)]}%
\]
}, respectively $\mathbf{e^{r}}\in\mathcal{C}\ell_{\mathrm{Spin}_{1,3}^{e}%
}^{r}(M,\mathtt{g})$, and let $I(M,\mathtt{g})$ be the subbundle of
$\mathcal{C}\ell_{\mathrm{Spin}_{1,3}^{e}}^{l}(M,\mathtt{g})$ generated by the
idempotent, that is, if $\mathbf{\Psi}$ is a section of $I(M,\mathtt{g}%
)\subset\mathcal{C}\ell_{\mathrm{Spin}_{1,3}^{e}}^{l}(M,\mathtt{g})$, we have%

\begin{equation}
\mathbf{\Psi e}=\mathbf{\Psi},
\end{equation}

A section $\mathbf{\Psi}$ of $I(M,\mathtt{g})$ is called a left ideal
algebraic spinor field.
\end{definition}

\begin{definition}
A Dirac-Hestenes spinor field \emph{(DHSF)} associated with $\mathbf{\Psi}$ is
a section\footnote{$\mathcal{C}\ell_{\mathrm{Spin}_{1,3}^{e}}^{0l}%
(M,\mathtt{g})$ denotes the even subbundle of $\mathcal{C}\ell_{\mathrm{Spin}%
_{1,3}^{e}}^{l}(M,\mathtt{g})$} $\varPsi$ of $\mathcal{C}\ell_{\mathrm{Spin}%
_{1,3}^{e}}^{0l}(M,\mathtt{g})\subset\mathcal{C}\ell_{\mathrm{Spin}_{1,3}^{e}%
}^{l}(M,\mathtt{g})$ such that\footnote{For any $\mathbf{\Psi}$ the DHSF
always exist, see \cite{rc2007}.}
\begin{equation}
\mathbf{\Psi}=\varPsi\mathbf{e}. \label{DHSF}%
\end{equation}

\end{definition}

\begin{definition}
We denote the complexified left spin-Clifford bundle by%
\[
\mathbb{C}\ell_{\mathrm{Spin}_{1,3}^{e}}^{l}(M,\mathtt{g})=P_{\mathrm{Spin}%
_{1,3}^{e}}(M,\mathtt{g})\times_{l}\mathbb{C}\otimes\mathbb{R}_{1,3}\equiv
P_{\mathrm{Spin}_{1,3}^{e}}(M,\mathtt{g})\times_{l}\mathbb{R}_{1,4}.
\]

\end{definition}

\begin{definition}
An equivalent definition of a \emph{DHSF} is the following. Let $\mathbf{\Psi
}\in\sec\mathbb{C}\ell_{\mathrm{Spin}_{1,3}^{e}}^{l}(M,\mathtt{g})$ such that%
\[
\mathbf{\Psi f=\Psi.}%
\]
Then a \emph{DHSF} associated with $\mathbf{\Psi}$ is an even section
$\varPsi$ of $\mathcal{C}\ell_{\mathrm{Spin}_{1,3}^{e}}^{0l}(M,\mathtt{g}%
)\subset\mathcal{C}\ell_{\mathrm{Spin}_{1,3}^{e}}^{l}(M,\mathtt{g})$ such that%
\begin{equation}
\mathbf{\Psi}=\varPsi\mathbf{f}.
\end{equation}

\end{definition}

\begin{definition}
There are natural pairings:%
\begin{align}
\sec\mathcal{C}\ell_{\mathrm{Spin}_{1,3}^{e}}^{l}(M,\mathtt{g})\times
\sec\mathcal{C}\ell_{\mathrm{Spin}_{1,3}^{e}}^{r}(M,\mathtt{g})  &
\rightarrow\sec\mathcal{C}\ell(M,\mathtt{g}),\label{PAIRING}\\
\sec\mathcal{C}\ell_{\mathrm{Spin}_{1,3}^{e}}^{r}(M,\mathtt{g})\times
\sec\mathcal{C}\ell_{\mathrm{Spin}_{1,3}^{e}}^{l}(M,\mathtt{g})  &
\rightarrow\mathcal{F(}M,\mathbb{R}_{1,3}),
\end{align}
such that given a section $\alpha$ of $\mathcal{C}\ell_{\mathrm{Spin}%
_{1,3}^{e}}^{l}(M,\mathtt{g})$ and a section $\beta$ of $\mathcal{C}%
\ell_{\mathrm{Spin}_{1,3}^{e}}^{r}(M,\mathtt{g})$ and selecting
representatives $(p,a)$ for $\alpha(x)$ and $(p,b)$ for $\beta(x)$
\emph{(}$p\in\pi^{-1}\left(  x\right)  $) it is%
\begin{align}
(\alpha\beta)  &  :=[(p;ab)]\in\mathcal{C}\ell(M,\mathtt{g}),\\
(\beta\alpha)(x)  &  :=ba\in\mathbb{R}_{1,3}.
\end{align}
If alternative representatives $(pu^{-1},ua)$ and $(pu^{-1},bu^{-1})$ are
chosen for $\alpha(x)$ and $\beta(x)$ we have $[(pu^{-1};uabu^{-1})]$, that,
by \emph{Eq.(\ref{eq_cliff})} represents the same element on $\mathcal{C}%
\ell(M,\mathtt{g})$, and $(bu^{-1}ua)=ba$; thus $(\alpha\beta)(x)$ and
$(\beta\alpha)(x)$ are a well defined. Following the same procedure we can
define the actions:%
\begin{align}
\sec\mathcal{C}\ell(M,\mathtt{g})\times\sec\mathcal{C}\ell_{\mathrm{Spin}%
_{1,3}^{e}}^{l}(M,\mathtt{g})  &  \rightarrow\sec\mathcal{C}\ell
_{\mathrm{Spin}_{1,3}^{e}}^{l}(M,\mathtt{g}),\\
\sec\mathcal{C}\ell_{\mathrm{Spin}_{1,3}^{e}}^{r}(M,\mathtt{g})\times
\sec\mathcal{C}\ell(M,\mathtt{g})  &  \rightarrow\sec\mathcal{C}%
\ell_{\mathrm{Spin}_{1,3}^{e}}^{r}(M,\mathtt{g}),\\
\sec\mathcal{C}\ell_{\mathrm{Spin}_{1,3}^{e}}^{l}(M,\mathtt{g})\times
\mathbb{R}_{1,3}  &  \rightarrow\sec\mathcal{C}\ell_{\mathrm{Spin}_{1,3}^{e}%
}^{l}(M,\mathtt{g}),\\
\mathbb{R}_{1,3}\times\sec\mathcal{C}\ell_{\mathrm{Spin}_{1,3}^{e}}%
^{r}(M,\mathtt{g})  &  \rightarrow\sec\mathcal{C}\ell_{\mathrm{Spin}_{1,3}%
^{e}}^{r}(M,\mathtt{g}).
\end{align}

\end{definition}

Given a local trivialization of $\mathcal{C}\ell(M,\mathtt{g})$ (or
$\mathcal{C}\ell_{\mathrm{Spin}_{1,3}^{e}}^{l}(M,\mathtt{g})$, $\mathcal{C}%
\ell_{\mathrm{Spin}_{1,3}^{e}}^{r}(M,\mathtt{g}),$ $\mathbb{C}\ell
_{\mathrm{Spin}_{p,q}^{e}}^{l}(M,\mathtt{g})$) ($\mathcal{U}\subset M$)%
\begin{equation}
\phi_{U}:\pi^{-1}(\mathcal{U})\rightarrow\mathcal{U}\times\mathbb{R}_{1,3},
\end{equation}
we can define a local unit section by $\mathbf{1}_{U}(x)=\phi_{U}^{-1}(x,1)$.
For $\mathcal{C}\ell(M,\mathtt{g})$, it is easy to show that a global unit
section always exist, independently of the fact that $M$ \ is parallelizable
or not. For the bundles $\mathcal{C}\ell_{\mathrm{Spin}_{p,q}^{e}}%
^{l}(M,\mathtt{g})$, $\mathcal{C}\ell_{\mathrm{Spin}_{p,q}^{e}}^{r}%
(M,\mathtt{g})$, $\mathbb{C}\ell_{\mathrm{Spin}_{p,q}^{e}}^{l}(M,\mathtt{g})$
($\dim$ $M=p+q$) there exist a global unit sections if, and only if,
$P_{\mathrm{Spin}_{p,q}^{e}}(M,\mathtt{g})$ is trivial
\cite{r2004,mr02004,rc2007}. In our case we know, that $M$ is parallelizable
and we can define global unit sections on $\mathcal{C}\ell_{\mathrm{Spin}%
_{1,3}^{e}}^{l}(M,\mathtt{g})$, $\mathcal{C}\ell_{\mathrm{Spin}_{1,3}^{e}}%
^{r}(M,\mathtt{g})$ and $\mathbb{C}\ell_{\mathrm{Spin}_{1,3}^{e}}%
^{r}(M,\mathtt{g})$.

Let $\boldsymbol{\mathbf{\Xi}}$$_{u}$ be a section of $P_{\mathrm{Spin}%
_{1,3}^{e}}(M,\mathtt{g})$, i.e., a spin frame. We recall, in order to fix
notations, that sections of $\mathcal{C}\ell(M,\mathtt{g})$ $I^{l}%
(M,\mathtt{g})$, $\mathcal{C}\ell_{\mathrm{Spin}_{1,3}^{e}}^{l}(M,\mathtt{g}%
)$, are, respectively, the equivalence classes%
\begin{gather}
\boldsymbol{C}=[(\boldsymbol{\mathbf{\Xi}}_{u},\mathcal{C}%
_{\boldsymbol{\mathbf{\Xi}}_{u}})],\nonumber\\
\mathbf{\Psi}=[(\boldsymbol{\mathbf{\Xi}}_{u},\boldsymbol{\Psi}%
_{\boldsymbol{\mathbf{\Xi}}_{u}})],~~~\varPsi=[(\boldsymbol{\mathbf{\Xi}}%
_{u},\varPsi_{\boldsymbol{\mathbf{\Xi}}_{u}})]. \label{notation}%
\end{gather}

\begin{remark}
When convenient, we will write $\mathcal{C}_{\Xi_{u}}\in\sec\mathcal{C}%
\ell(M,\mathtt{g})$ to mean that there exists a section $\boldsymbol{C}$ of
the Clifford bundle $\mathcal{C}\ell(M,\mathtt{g})$ defined by
$[(\boldsymbol{\mathbf{\Xi}}_{u},\mathcal{C}_{\boldsymbol{\mathbf{\Xi}}_{u}%
})]$. Analogous notations will be used for sections of the other bundles
introduced above. Also, when there is no chance of confusion on the chosen
spinor frame, we will write $\mathcal{C}_{\boldsymbol{\mathbf{\Xi}}_{u}}$
simply as $\mathcal{C}$.
\end{remark}

For each spin frame, say $\boldsymbol{\mathbf{\Xi}}$$_{0}$, let $\mathbf{1}%
_{\boldsymbol{\mathbf{\Xi}}_{0}}^{l}$ and $\mathbf{1}_{\boldsymbol{\mathbf{\Xi
}}_{0}}^{r}$ be the global unit sections of $\mathcal{C\ell}_{\mathrm{Spin}%
_{1,3}^{e}}^{l}(M,\mathtt{g})$ and $\mathcal{C\ell}_{\mathrm{Spin}_{1,3}^{e}%
}^{r}(M,\mathtt{g})$, given by%

\begin{equation}
\mathbf{1}_{\boldsymbol{\mathbf{\Xi}}_{0}}^{r}:=[(\Xi_{0},1)],~~~~~\mathbf{1}%
_{\boldsymbol{\mathbf{\Xi}}_{0}}^{l}:=[(\Xi_{0},1)]. \label{repre3}%
\end{equation}

\begin{remark}
Before proceeding note that given another spin frame $\Xi_{u}=\Xi_{0}u$, where
$u:M\rightarrow\mathrm{Spin}_{1,3}^{e}\subset\mathbb{R}_{1,3}^{0}%
\subset\mathbb{R}_{1,3}$ we define the sections $\mathbf{1}%
_{\boldsymbol{\mathbf{\Xi}}_{u}}^{r}$ of $\mathcal{C\ell}_{\mathrm{Spin}%
_{1,3}^{e}}^{r}(M,\mathtt{g})$ and $\mathbf{1}_{\boldsymbol{\mathbf{\Xi}}_{0}%
}^{l}$ of $\mathcal{C\ell}_{\mathrm{Spin}_{1,3}^{e}}^{l}(M,\mathtt{g})$ by
\begin{equation}
\mathbf{1}_{\boldsymbol{\mathbf{\Xi}}_{u}}^{r}:=[(\Xi_{u},1)],~~~~~\mathbf{1}%
_{\boldsymbol{\mathbf{\Xi}}_{u}}^{l}:=[(\Xi_{u},1)]. \label{rep4}%
\end{equation}

It has been proved in \emph{\cite{r2004,rc2007}} that the relation between
$\mathbf{1}_{\boldsymbol{\mathbf{\Xi}}}^{r}$ and $\mathbf{1}%
_{\boldsymbol{\mathbf{\Xi}}_{0}}^{r}$ and between $\mathbf{1}%
_{\boldsymbol{\mathbf{\Xi}}}^{l}$ and $\mathbf{1}_{\boldsymbol{\mathbf{\Xi}%
}_{0}}^{l}$\ are given by
\begin{equation}
\mathbf{1}_{\boldsymbol{\mathbf{\Xi}}_{u}}^{r}=u^{-1}\mathbf{1}%
_{\boldsymbol{\mathbf{\Xi}}_{0}}^{r}=\mathbf{1}_{\boldsymbol{\mathbf{\Xi}}%
_{0}}^{r}U^{-1},~~~~~\mathbf{1}_{\boldsymbol{\mathbf{\Xi}}_{u}}^{l}%
=U\mathbf{1}_{\boldsymbol{\mathbf{\Xi}}_{0}}^{l}=\mathbf{1}%
_{\boldsymbol{\mathbf{\Xi}}_{0}}^{l}u \label{rep5}%
\end{equation}
where $U$ is the section of $\mathcal{C}\ell(M,\mathtt{g})$ defined by the
equivalence class
\begin{equation}
U=[(\Xi_{0},u)]. \label{rep6}%
\end{equation}

\end{remark}

The unity sections $\mathbf{1}_{\boldsymbol{\mathbf{\Xi}}_{u}}^{l}$ and
$\mathbf{1}_{\boldsymbol{\mathbf{\Xi}}_{u}}^{r}$satisfies the important
relations\footnote{$\mathcal{C}\ell^{0}(M,\mathtt{g})$ denotes the even
subbundle of $\mathcal{C}\ell(M,\mathtt{g})$.}%

\begin{equation}
\mathbf{1}_{\boldsymbol{\mathbf{\Xi}}_{u}}^{l}\mathbf{1}%
_{\boldsymbol{\mathbf{\Xi}}_{u}}^{r}=1\in\sec\mathcal{C}\ell(M,\mathtt{g}%
),~\mathbf{1}_{\boldsymbol{\mathbf{\Xi}}_{u}}^{r}\mathbf{1}%
_{\boldsymbol{\mathbf{\Xi}}_{u}}^{l}=1\in\mathcal{F(}M,\mathbb{R}_{1,3}),
\label{rep6a}%
\end{equation}

\begin{definition}
A representative of a DHSF $\varPsi$ in the Clifford bundle $\mathcal{C}%
\ell(M,\mathtt{g})$ relative to a spin frame $\mathbf{\Xi}_{u}$ is a section
$\boldsymbol{\psi}_{\mathbf{\Xi}_{u}}=[(\mathbf{\Xi}_{u},\psi_{\mathbf{\Xi
}_{u}})]$ of $\mathcal{C}\ell^{0}(M,\mathtt{g})$\emph{ }given by
\emph{\cite{r2004,mr02004,rc2007}}
\end{definition}

\begin{equation}
\boldsymbol{\psi}_{\mathbf{\Xi}_{u}}=\Psi\mathbf{1}_{\mathbf{\Xi}_{u}}^{r},.
\label{RDHSF}%
\end{equation}
Representatives in the Clifford bundle of $\varPsi$ relative to spin frames,
say $\mathbf{\Xi}_{u^{\prime}}$ and $\mathbf{\Xi}_{u}$, are related
by\footnote{This relation has been used in \cite{r2004} to define a DHSF as an
appropriate equivalence class of even sections of the Clifford bundle
$\mathcal{C\ell}(M,\mathtt{g})$.}%

\begin{equation}
\boldsymbol{\psi}_{\mathbf{\Xi}_{u^{\prime}}}U^{\prime-1}=\boldsymbol{\psi
}_{\mathbf{\Xi}_{u}}U^{-1}. \label{30}%
\end{equation}

In the main text we use the symbol $\phi$ as a short for the representative of
a DHSF in the spinor basis defined by the fiducial frame $\Xi_{0}$.

\textbf{DHSFs} unveil the hidden geometrical meaning of spinors (and spinor
fields). Indeed, consider $v\in\mathbb{R}^{1,3}\hookrightarrow\mathbb{R}%
_{1,3}$ be, initially, a timelike covector\ such that $v^{2}=1.$ The linear
mapping, belonging to $\mathrm{SO}_{1,3}^{e}$%
\begin{equation}
v\mapsto RvR^{-1}=Rv\tilde{R}=w,~R\in\mathrm{Spin}_{1,3}^{e},
\end{equation}
define a new covector $w$ such that $w^{2}=1.$ We can therefore fix a covector
$v$ and obtain all other unit timelike covectors by applying this mapping.
This same procedure can be generalized to obtain any type of timelike covector
starting from a fixed unit covector $v$. We define the linear mapping%
\begin{equation}
v\mapsto\psi v\tilde{\psi}=z \label{rotate}%
\end{equation}
to obtain $z^{2}=\rho^{2}>0$. Since $z$ can be written as $z=\rho Rv\tilde{R}%
$, we need%
\begin{equation}
\psi v\tilde{\psi}=\rho Rv\tilde{R}. \label{ddhsf}%
\end{equation}
If we write \ $\psi=\rho^{\frac{1}{2}}MR$ we need that $Mv\tilde{M}=v$ and the
most general solution is $M=e^{\frac{\tau_{\boldsymbol{g}}\beta}{2}}$, where
$\tau_{\boldsymbol{g}}=\gamma^{0}\gamma^{1}\gamma^{2}\gamma^{3}\in%
{\textstyle\bigwedge\nolimits^{4}}
\mathbb{R}^{1,3}\hookrightarrow\mathbb{R}_{1,3}$ and $\beta\in\mathbb{R}$ is
called the Takabayasi angle \cite{rc2007,jayme}. Then follows that $\psi$ is
of the form%
\begin{equation}
\psi=\rho^{\frac{1}{2}}e^{\frac{\tau_{\boldsymbol{g}}\beta}{2}}R.
\label{takabaiasi}%
\end{equation}

Now, Eq.(\ref{takabaiasi}) shows that $\psi\in\mathbb{R}_{1,3}^{0}%
\simeq\mathbb{R}_{3,0}$. Moreover, we have that $\psi\tilde{\psi}\neq0$ since%
\begin{equation}
\psi\tilde{\psi}=\rho e^{\tau_{\boldsymbol{g}}\beta}=(\rho\cos\beta
)+\tau_{\boldsymbol{g}}(\rho\sin\beta).
\end{equation}

A representative of a DHSF $\varPsi$ in the Clifford bundle $\mathcal{C}%
\ell(M,\mathtt{g})$ relative to a spin frame $\mathbf{\Xi}_{u}$ is a section
$\boldsymbol{\psi}_{\mathbf{\Xi}_{u}}=[(\mathbf{\Xi}_{u},\psi_{\mathbf{\Xi
}_{u}})]$ of $\mathcal{C}\ell^{0}(M,\mathtt{g})$ where $\psi_{\mathbf{\Xi}%
_{u}}\in\mathbb{R}_{1,3}^{0}\simeq\mathbb{R}_{3,0}$. So a DHSF such
$\boldsymbol{\psi}_{\mathbf{\Xi}_{u}}\boldsymbol{\tilde{\psi}}_{\mathbf{\Xi
}_{u}}\neq0$ induces a linear mapping induced by Eq.(\ref{rotate}), which
\emph{rotates} a covector field and \emph{dilate} it.

\section{Description of the Dirac Equation in the Clifford Bundle}

To fix the notation let $(M\simeq\mathbb{R}^{4},\boldsymbol{\eta}%
,D,\tau_{\mathbf{\eta}})$ be the Minkowski spacetime structure where
$\boldsymbol{\eta}\in\sec T_{0}^{2}M$ is Minkowski metric and $D$ is the
Levi-Civita connection of $\boldsymbol{\eta}$. Also, $\tau_{\mathbf{\eta}}%
\in\sec%
{\textstyle\bigwedge\nolimits^{4}}
T^{\ast}M$ defines an orientation. We denote by $\mathtt{\eta}\in\sec
T_{2}^{0}M$ the metric of the cotangent bundle. It is defined as follows. Let
$\{x^{\mu}\}$ be coordinates for\ $M$ in the Einstein-Lorentz-Poincar\'{e}
gauge \cite{rc2007}. Let $\{\boldsymbol{e}_{\mu}=\partial/\partial x^{\mu}\}$
a basis for $TM$ and $\{\gamma^{\mu}=dx^{\mu}\}$ the corresponding dual basis
for $T^{\ast}M$, i.e., $\gamma^{\mu}(\boldsymbol{e}_{\alpha})=\delta_{\alpha
}^{\mu}$. Then, if $\boldsymbol{\eta}=\eta_{\mu\nu}\gamma^{\mu}\otimes
\gamma^{\nu}$ then $\mathtt{\eta}=\eta^{\mu\nu}\boldsymbol{e}_{\mu}%
\otimes\boldsymbol{e}_{\nu}$, where the matrix with entries $\eta_{\mu\nu}$
and the one with entries $\eta^{\mu\nu}$ are the equal to the diagonal matrix
$\mathrm{diag}(1,-1,-1,-1)$. If $a,b\in\sec%
{\textstyle\bigwedge\nolimits^{1}}
T^{\ast}M$\ we write $a\cdot b=\mathtt{\eta}(a,b)$. We also denote by
$\langle\gamma_{\mu}\rangle$ the reciprocal basis of $\{\gamma^{\mu}=dx^{\mu
}\}$, which satisfies $\gamma^{\mu}\cdot\gamma_{\nu}=\delta_{\nu}^{\mu}$.

We denote the Clifford bundle of differential forms\footnote{We recall that
$\mathcal{C\ell}(T_{x}^{\ast}M,\eta)\simeq\mathbb{R}_{1,3}$ the so-called
spacetime algebra. Also the even subalgebra of $\mathbb{R}_{1,3}$ denoted
$\mathbb{R}_{1,3}^{0}$ is isomorphic to te Pauli algebra $\mathbb{R}_{3,0}$,
i.e., $\mathbb{R}_{1,3}^{0}\simeq\mathbb{R}_{3,0}$. The even subalgebra of the
Pauli algebra $\mathbb{R}_{3,0}^{0}:=\mathbb{R}_{3,0}^{00}$ is the quaternion
algebra $\mathbb{R}_{0,2}$, i.e., $\mathbb{R}_{0,2}\simeq\mathbb{R}_{3,0}^{0}%
$. Moreover we have the identifications: $\mathrm{Spin}_{1,3}^{0}%
\simeq\mathrm{Sl}(2,\mathbb{C})$, $\mathrm{Spin}_{3,0}\simeq\mathrm{SU}(2)$.
For the Lie algebras of these groups we have $\mathrm{spin}_{1,3}^{0}%
\simeq\mathrm{sl}(2,\mathbb{C})$,$\ \mathrm{su}(2)\simeq\mathrm{spin}_{3,0}$.
The important fact to keep in mind for the understanding of some of the
identificastions we done below is that $\mathrm{Spin}_{1,3}^{0},\mathrm{spin}%
_{1,3}^{0}\subset\mathbb{R}_{3,0}\subset\mathbb{R}_{1,3}$ and $\mathrm{Spin}%
_{3,0},\mathrm{spin}_{3,0}\subset\mathbb{R}_{0,2}\subset\mathbb{R}_{1,3}%
^{0}\subset\mathbb{R}_{1,3}$.} in Minkowski spacetime by $\mathcal{C\ell
}(M,\eta)$ and use notations and conventions in what follows as in
\cite{rc2007} and recall the fundamental relation%
\begin{equation}
\gamma^{\mu}\gamma^{\nu}+\gamma^{\nu}\gamma^{\mu}=2\eta^{\mu\nu}. \label{1}%
\end{equation}

If $\{\boldsymbol{\gamma}^{\mu},~$ $\mu=0,1,2,3\}$ are the Dirac gamma
matrices in the \emph{standard representation} and $\{\gamma_{\mu}%
,~\mu=0,1,2,3\}$ are as introduced above, we define%
\begin{align}
\sigma_{k}  &  :=\gamma_{k}\gamma_{0}\in\sec%
{\textstyle\bigwedge\nolimits^{2}}
T^{\ast}M\hookrightarrow\sec\mathcal{C\ell}^{0}(M,\eta)\text{, }%
k=1,2,3,\label{2}\\
\mathbf{i}  &  =\gamma_{5}:=\gamma_{0}\gamma_{1}\gamma_{2}\gamma_{3}\in\sec%
{\textstyle\bigwedge\nolimits^{4}}
T^{\ast}M\hookrightarrow\sec\mathcal{C\ell}(M,\eta),\\
\boldsymbol{\gamma}_{5}  &  :=\boldsymbol{\gamma}_{0}\boldsymbol{\gamma}%
_{1}\boldsymbol{\gamma}_{2}\boldsymbol{\gamma}_{3}\in\mathbb{C(}4\mathbb{)}%
\end{align}

Noting that $M$ is parallelizable, in a given global spin frame a covariant
spinor field can be taken as a mapping $\boldsymbol{\psi}:M\rightarrow
\mathbb{C}^{4}$ \ In standard representation of the gamma matrices where
($i=\sqrt{-1}$, $\boldsymbol{\phi},\boldsymbol{\varsigma}:M\rightarrow
\mathbb{C}^{2}$) to $\boldsymbol{\psi}$ given by
\begin{equation}
\boldsymbol{\psi}=\left(
\begin{array}
[c]{c}%
\boldsymbol{\phi}\\
\boldsymbol{\varsigma}%
\end{array}
\right)  =\left(
\begin{array}
[c]{c}%
\left(
\begin{array}
[c]{c}%
m^{0}+im^{3}\\
-m^{2}+im^{1}%
\end{array}
\right) \\
\left(
\begin{array}
[c]{c}%
n^{0}+in^{3}\\
-n^{2}+in^{1}%
\end{array}
\right)
\end{array}
\right)  , \label{3}%
\end{equation}
there corresponds the \textbf{DHSF} $\psi\in\sec\mathcal{C\ell}^{0}(M,\eta)$
given by\footnote{Remember the identification:%
\[
\mathbb{C}(4)\simeq\mathbb{R}_{4,1}\supseteq\mathbb{R}_{4,1}^{0}%
\simeq\mathbb{R}_{1,3}.
\]
}%
\begin{equation}
\psi=\phi+\varsigma\sigma_{3}=(m^{0}+m^{k}\mathbf{i}\sigma_{k})+(n^{0}%
+n^{k}\mathbf{i}\sigma_{k})\sigma_{3}. \label{4}%
\end{equation}
We then have the useful formulas in Eq.(\ref{5}) below that one can use to
immediately translate results of the standard matrix formalism in the language
of the Clifford bundle formalism and vice-versa\footnote{$\tilde{\psi}$ is the
reverse of $\psi$. If $A_{r}\in\sec%
{\textstyle\bigwedge\nolimits^{r}}
T^{\ast}M\hookrightarrow\sec\mathcal{C\ell}(M,\eta)$ then $\tilde{A}%
_{r}=(-1)^{\frac{r}{2}(r-1)}A_{r}$.}
\begin{align}
\boldsymbol{\gamma}_{\mu}\boldsymbol{\psi}  &  \leftrightarrow\gamma_{\mu}%
\psi\gamma_{0},\nonumber\\
i\boldsymbol{\psi}  &  \leftrightarrow\psi\gamma_{21}=\psi\mathbf{i}\sigma
_{3},\nonumber\\
i\boldsymbol{\gamma}_{5}\boldsymbol{\psi}  &  \leftrightarrow\psi\sigma
_{3}=\psi\gamma_{3}\gamma_{0},\nonumber\\
\boldsymbol{\bar{\psi}}  &  =\boldsymbol{\psi}^{\dagger}\boldsymbol{\gamma
}^{0}\leftrightarrow\tilde{\psi},\nonumber\\
\boldsymbol{\psi}^{\dagger}  &  \leftrightarrow\gamma_{0}\tilde{\psi}%
\gamma_{0},\nonumber\\
\boldsymbol{\psi}^{\ast}  &  \leftrightarrow-\gamma_{2}\psi\gamma_{2}.
\label{5}%
\end{align}

Using the above dictionary the standard Dirac equation\footnote{$\partial
_{\mu}:=\frac{\partial}{\partial x^{\mu}}$.} for a Dirac spinor field
$\boldsymbol{\psi}:M\rightarrow\mathbb{C}^{4}$
\begin{equation}
i\boldsymbol{\gamma}^{\mu}\partial_{\mu}\boldsymbol{\psi}-m\boldsymbol{\psi
}=0\label{6}%
\end{equation}
translates immediately in the so-called Dirac-Hestenes equation, i.e.,
\begin{equation}
\boldsymbol{\partial}\psi\gamma_{21}-m\psi\gamma_{0}=0.\label{7}%
\end{equation}

\section{Generalized Dirac-Hestenes Spinors for $\mathbb{R}_{4,1}$}

Let $\{\boldsymbol{E}^{A}\}$, $a=1,2,3,4,0$ be an orthonormal basis for
$\mathbb{R}_{4,1}$. We have
\begin{equation}
\boldsymbol{E}^{A}\boldsymbol{E}^{B}+\boldsymbol{E}^{B}\boldsymbol{E}%
^{A}=2\eta^{AB}\label{411}%
\end{equation}
where the matrix with entries $\eta^{AB}$ is the diagonal matrix
\textrm{diag}$(1,1,1,1,-1).$Define $\mathfrak{i:=}\boldsymbol{E}%
^{0}\boldsymbol{E}^{1}\boldsymbol{E}^{2}\boldsymbol{E}^{3}\boldsymbol{E}^{4}$
and $\Gamma^{\mu}=\boldsymbol{E}^{\mu}\boldsymbol{E}^{4}$. Then,
\begin{equation}
\Gamma^{\mu}\Gamma^{\nu}+\Gamma^{\nu}\Gamma^{\mu}=2\eta^{\mu\nu}\label{412}%
\end{equation}
where the matrix with entries $\eta^{\mu\nu}$ is the diagonal matrix
\textrm{diag}$(1,-1,-1,-1)$.\ Recalling that $\mathbb{R}_{4,1}\simeq
\mathbb{C\otimes R}_{1,3}\simeq\mathbb{C\otimes R}_{3,1}$ and that
\begin{equation}
f=\frac{1}{2}(1+\Gamma^{0})\frac{1}{2}(1+\mathfrak{i}\Gamma^{2}\Gamma
^{1})\label{413}%
\end{equation}
is a primitive idempotent of $\mathbb{R}_{4,1}$ and $I=\mathbb{R}_{4,1}f$ is a
minimum ideal of $\mathbb{R}_{4,1}$ such that $\dim_{\mathbb{R}}I=8$. Let
$\Psi\in\mathbb{R}_{4,1}f$ and $Z\in\mathbb{R}_{4,1}^{0}\frac{1}{2}%
(1+\Gamma^{0})$. Since $\mathbb{R}_{4,1}^{0}\simeq\mathbb{R}_{1,3}$ we
recognize $Z$ as isomorphic to some covariant Dirac spinor. We can easily find
the following relation between $\Psi$ and $Z$%
\begin{equation}
\Psi=Z\frac{1}{2}(1+\mathfrak{i}\Gamma^{2}\Gamma^{1}).\label{414}%
\end{equation}
Moreover, decomposing $Z$ into even and odd parts relative to the
$\mathbb{Z}_{2}$-gradation of $\mathbb{R}_{4,1}^{0}\simeq\mathbb{R}_{1,3}$,
$Z=Z^{0}+Z^{1}$ we find that $Z^{1}=Z^{0}\Gamma^{0}$, which clearly shows that
all information in $Z$ is contained in $Z^{0}$. Then, we have
\begin{equation}
\Psi=Z^{0}(1+\Gamma^{0})\frac{1}{2}(1+\mathfrak{i}\Gamma^{2}\Gamma
^{1})\label{414a}%
\end{equation}

Now, taking into account the well known result that
\begin{equation}
\mathbb{R}_{4,1}^{0}\frac{1}{2}(1+\Gamma^{0})=\mathbb{R}_{4,1}^{00}\frac{1}%
{2}(1+\Gamma^{0})\label{415}%
\end{equation}
where $\mathbb{R}_{4,1}^{00}\simeq\mathbb{R}_{1,3}^{0}$ is the even subalgebra
of $\mathbb{R}_{4,1}^{0}\simeq\mathbb{R}_{1,3}$ we see that $Z\in
\mathbb{R}_{4,1}^{0}\frac{1}{2}(1+\Gamma^{0})$ can be written as%
\begin{equation}
Z=\psi\frac{1}{2}(1+\Gamma^{0}),~~~~\label{416}%
\end{equation}
with $\psi\in\mathbb{R}_{4,1}^{00}\simeq\mathbb{R}_{1,3}^{0}$ a representative
of a Dirac-Hestenes spinor. Thus, putting $Z^{0}=\psi/2$ we end with the
notable expression\footnote{If you need more details consult Section 3.4 of
\cite{rc2007}.}%
\begin{equation}
\Psi=\psi\frac{1}{2}(1+\Gamma^{0})\frac{1}{2}(1+\mathfrak{i}\Gamma^{2}%
\Gamma^{1}).\label{417}%
\end{equation}

Next note that Hestenes (see page 106 of \cite{hs1984}) defines a spinor for
$\mathbb{R}_{4,1}$the object
\begin{equation}
\boldsymbol{\Phi}=\zeta^{1/2}\mathbf{V},\label{418}%
\end{equation}
where $\zeta$ is a scalar and $\mathbf{V}\in\mathbb{R}_{4,1}^{0}$ and
$\mathbf{V\tilde{V}}=1$ such that for any $\mathbf{x}\in\mathbb{R}_{4,1}$,
$\Upsilon\mathbf{x}\Upsilon\in\mathbb{R}_{4,1}$. To relate $\Upsilon$ with
$\psi$ we write
\begin{equation}
\boldsymbol{\Phi=}\psi\frac{1}{2}(1+\Gamma^{0})+\mathbf{\xi}^{\alpha}%
\Gamma_{0}\Gamma_{\alpha}\in\mathbb{R}_{4,1}^{0}.\label{419}%
\end{equation}

\begin{remark}
\label{XX}Take notice that the $\mathbf{\xi}^{\alpha}$ are four scalars which
must satisfy some link in order to define the additional $3$ degrees of
freedom of $\boldsymbol{\Phi}$ and $\psi\in\mathbb{R}_{4,1}^{00}%
\simeq\mathbb{R}_{1,3}^{0}$ such that $\psi\tilde{\psi}\neq0$ is written as
\begin{equation}
\psi=\mathbf{S\ +}\frac{1}{2}\mathbf{B}_{\mu\nu}\Gamma^{\mu}\Gamma^{\nu
}+\mathbf{P}\Gamma^{0}\Gamma^{1}\Gamma^{2}\Gamma^{3}=\mathbf{\rho}%
^{1/2}e^{-\frac{\beta\Gamma^{0}\Gamma^{1}\Gamma^{2}\Gamma^{3}}{2}}%
\mathbf{U}\label{420}%
\end{equation}
with $\mathbf{U\in}\mathbb{R}_{1,3}^{0}$ such that $\mathbf{U\tilde{U}=1}$.
Moreover, since $\boldsymbol{\Phi\tilde{\Phi}}=\zeta$ we need the additional
link
\begin{equation}
\left[  \psi\frac{1}{2}(1+\Gamma^{0})+\mathbf{\xi}^{\alpha}\Gamma_{0}%
\Gamma_{\alpha}\right]  \left[  \frac{1}{2}\tilde{\psi}(1-\Gamma
^{0})-\mathbf{\xi}^{\alpha}\Gamma_{0}\Gamma_{\alpha}\right]  =\zeta
.\label{421}%
\end{equation}

\end{remark}

\section{Heuristic Derivation of the DHE in Minkowski Spacetime}

We start recalling that a classical spin $1/2$ free particle is supposed to
have its story described by a geodesic timelike worldline $\sigma
:\mathbb{R}\supset I\rightarrow M(\approx\mathbb{R}^{4})$ in the Minkowski
spacetime structure Let $\sigma_{\ast}$ be the velocity of the particle and
let $\boldsymbol{v}=\boldsymbol{\eta}(\sigma_{\ast},)$ be the
\textit{physically equivalent} 1-form. We that $\boldsymbol{v}\in\sec
T_{\sigma}^{\ast}M\hookrightarrow\sec\mathcal{C}{\ell}(M,\mathtt{\eta})$. Its
classical momentum $1$-form is%
\begin{equation}
\boldsymbol{p}=m\boldsymbol{v.} \label{ncdh0}%
\end{equation}

To continue, we suppose the existence of a 1-form field $V\in\sec\bigwedge
^{1}T^{\ast}M\hookrightarrow\sec\mathcal{C}{\ell}(M,\mathtt{g})$ such that its
restriction over $\sigma$ is $\boldsymbol{v}$, i.e., $V_{\mid\sigma
}=\boldsymbol{v}$. Also we impose that $V^{2}=1$. We introduce also the
$\boldsymbol{P}$ vector field such that $\left.  \boldsymbol{P}\right\vert
_{\sigma}=\boldsymbol{p}$ and consider the equation%
\begin{equation}
\boldsymbol{P}=m\boldsymbol{V} \label{ncdh01}%
\end{equation}

As in the previous appendix, let $\psi\in\sec\mathcal{C}{\ell}^{0}%
(M,\mathtt{g})$, be the representative (in the spin coframe $\Xi$) of a
\textit{particular} invertible Dirac-Hestenes spinor field such that%
\begin{equation}
\psi\tilde{\psi}\neq0 \label{ncdh5}%
\end{equation}
\ and since $\psi=\rho^{\frac{1}{2}}e^{\frac{\beta}{2}\gamma^{5}}R$ we have%
\begin{equation}
\psi\gamma^{0}\tilde{\psi}=\rho e^{\beta\gamma^{5}}R\gamma^{0}\tilde{R}.
\label{ncdh6}%
\end{equation}
which necessarily implies that if we want $\ V=\psi\gamma^{0}\tilde{\psi}$ we
need
\begin{equation}
e^{\beta\gamma^{5}}=\pm1, \label{ncdh00}%
\end{equation}
i.e., $\beta=0$ or $\beta=\pi$\footnote{These values correspond to charges of
opposite signs, see \cite{rc2007}, Chapter}. In what follows we take $\beta
=0$. Thus Eq.(\ref{ncdh01}) becomes%
\begin{equation}
\boldsymbol{P}=m\psi\gamma^{0}\tilde{\psi} \label{ncdh02}%
\end{equation}
and thus
\begin{equation}
\boldsymbol{P}\psi=m\psi\gamma^{0} \label{ncdh03}%
\end{equation}

Eq.(\ref{ncdh03}) is a purely classical equation which is simply another way
of writing Eq.(\ref{ncdh01}). To get a quantum mechanics \ wave equation we
must now change $\boldsymbol{P}$ into $\mathbf{P}$, the quantum mechanics
momentum operator. From the previous appendix we know that
\begin{equation}
\mathbf{P}\psi=\boldsymbol{\partial}\psi\gamma^{2}\gamma^{1}=\gamma^{\mu
}\partial_{\mu}\psi\gamma^{2}\gamma^{1}. \label{ncdh04}%
\end{equation}

Substituting this result in Eq.(\ref{ncdh03}) we get (compare with Eq.(23) of
\cite{roldao3})
\begin{equation}
\boldsymbol{\partial}\psi\gamma^{2}\gamma^{1}-m\psi\gamma^{0}=0.
\label{ncdh05}%
\end{equation}
which is the DHE, which, as well known, is completely equivalent to the
standard Dirac equation formulated in terms of covariant Dirac spinor fields.

\end{document}